%% file: main.tex
\documentclass[1p, authoryear]{elsarticle}

\usepackage{graphicx,mathrsfs,epstopdf}

\usepackage{amsmath,amsfonts,amssymb,amsthm}
\usepackage[colorlinks,citecolor=blue,urlcolor=blue]{hyperref}
\usepackage{algorithm} % has to be loaded AFTER hypperrf
\usepackage{algorithmic}
\usepackage{framed} % for highlighting paragraph
\usepackage{blkarray}
\usepackage[capitalize, sort]{cleveref}
\usepackage{makecell}
\usepackage{comment}

\usepackage{natbib}
\usepackage{enumerate}
\usepackage{subfigure}
\usepackage{tikz}
\usetikzlibrary{positioning}
\usepackage[utf8]{inputenc}
\usepackage{xcolor}
\usepackage{graphicx}
\usepackage{booktabs}
\usepackage[colorlinks]{hyperref}
\usepackage{mathtools}
\usepackage [english]{babel}
\usepackage [autostyle, english = american]{csquotes}
\MakeOuterQuote{"}

\newcommand{\R}{\mathbb{R}}

\newcommand{\ra}[1]{\textcolor{black}{#1}}
\newcommand{\Cov}{\mathrm{Cov}}
\newcommand{\mtp}{\text{MTP}_2}
\newcommand{\E}{\mathrm{E}}
\newcommand{\hatSigma}{\widehat{\Sigma}}

\DeclareMathOperator*{\argmax}{arg\,max}
\DeclareMathOperator*{\argmin}{arg\,min}

\newtheorem{thm}{Theorem}[section]
\newtheorem{lem}[thm]{Lemma}

\theoremstyle{definition}
\newtheorem{defn}[thm]{Definition}

\newtheorem{ex}[thm]{Example}

\theoremstyle{remark}

\title{Covariance Matrix Estimation under Total Positivity \\for Portfolio Selection}
\author[a,b,c]{Raj Agrawal} 
\author[b,c]{Uma Roy}
\author[b,c]{Caroline Uhler}

\address[a]{Computer Science and Artificial Intelligence Laboratory, Massachusetts Institute of Technology}
\address[b]{Laboratory for Information and Decision Systems, Massachusetts Institute of Technology}
\address[c]{Institute for Data, Systems and Society, Massachusetts Institute of Technology}

\begin{document}
\begin{abstract}
Selecting the optimal Markowitz porfolio depends on estimating the covariance matrix of the returns of $N$ assets from $T$ periods of historical data. Problematically, $N$ is typically of the same order as $T$, which makes the sample covariance matrix estimator perform poorly, both empirically and theoretically. While various other general purpose covariance matrix estimators have been introduced in the financial economics and statistics literature for dealing with the high dimensionality of this problem, we here propose an estimator that exploits the fact that assets are typically positively dependent. This is achieved by imposing that the joint  distribution of returns be \emph{multivariate totally positive of order 2} ($\mtp$). This constraint on the covariance matrix not only enforces positive dependence among the assets, but also regularizes the covariance matrix, leading to desirable statistical properties such as sparsity. Based on stock-market data spanning  thirty years, we show that estimating the covariance matrix under $\mtp$ outperforms previous state-of-the-art methods including shrinkage estimators and factor models. 
%there have been a number of papers both in the financial economics and statistical literature for improved high-dimensional covariance estimation, there has not been a paper that exploits the fact that the covariances between assets are typically positive. To this end, we introduce a new structural constraint, namely that the joint  distribution of returns be \emph{multivariate totally postive of order 2} ($\text{MTP}_2$). Such a constraint not only enforces that the covariance matrix contain all positive entries as desired, but also leads to desirable statistical properties such as \emph{sparsity} and \emph{regularization}. We show that covariance estimation under this constraint outperforms baseline methods (in terms of out-of-sample risk) based on \emph{shrinkage} techniques or explicitly provided \emph{market factors} on real stock-market data spanning over thirty years.      
\end{abstract}
\maketitle

%Overall story
% \begin{itemize}
%     \item How to pick a portfolio?
%     \item Markowitz
%     \item Requires estimating the true covariance matrix. Each entry of the sample covariance is estimated with effectively O(1) number of samples 
%     \item Current Solutions (add more structure in the true covariance matrix to make estimation easier : shrinkage estimator), non-linear shrinkage estimator, POET (PCA based approximate factor model)
%     \item Issues: unclear how to select good factors, performance is less good in the high-dimensional setting
%     \item Our solution: assumes additional structure of covariance matrix (totally positive, motivated by factor model interpretation of covariance matrix estimation) to give lower variance estimator (especially in high dimensional setting)
%     \item Advantage 1: in high dimensional setting, vanilla MTP2 outperforms linear shrinkage
% \end{itemize}

\section{Introduction} \label{sec:introduction}
\input{1_introduction.tex}

\section{Problem Statement} \label{sec:background}
\input{2_background.tex}

\section{Covariance Matrix Estimation under $\text{MTP}_2$} \label{sec:MTP2}
\input{3_MTP2.tex}

\section{Related Work} \label{sec:past_cov_estimations}
\input{related_work.tex}

\section{Empirical Evaluation} \label{sec:results}
\input{4_results.tex}

\section{Conclusion}
\input{5_conclusion.tex}

\section*{Acknowledgements}
The authors were partially supported by NSF (DMS-1651995), ONR (N00014-17-1-2147 and N00014-18-1-2765), IBM, a Sloan Fellowship and a Simons Investigator Award to C.~Uhler. We would like to thank Michael Wolf for helpful discussions and providing us with code for preprocessing the CRSP dataset. We would also like to thank the organizers and participants of the 11th Society of Financial Econometrics (SoFiE) Conference for their encouragement and helpful discussions that led to the pursuit of this project; in particular, we thank Christian Gourieroux, Robert F.~Engle, Christian Brownlees, Federico Bandi and Andrew Patton.

%\section*{References}

\bibliographystyle{agsm}
\bibliography{ref}

\clearpage

\appendix
\input{appendix.tex}

\end{document}

%% file: 1_introduction.tex
Given a universe of $N$ assets, what is the optimal way to select a portfolio? When "optimal" refers to selecting the portfolio with minimal risk or variance for a given level of expected return, then the solution, commonly known as the \emph{Markowitz optimal portfolio}, depends on two quantities: the vector of expected returns $\mu^*$ and the covariance matrix between returns $\Sigma^*$ \citep{markowitz}. In practice, $\mu^*$ and $\Sigma^*$ are unknown and must be estimated from historical returns. Since $\Sigma^*$ requires estimating $O(N^2)$ parameters while $\mu^*$ only requires estimating $O(N)$ parameters, the main challenge lies in estimating $\Sigma^*$. A naive strategy is to use the sample covariance matrix $S$ to estimate $\Sigma^*$. However, this estimator is known to have poor properties \citep{marcenko_dist,wachter1978_issue_samp_cov,bai1993_issue_samp_cov,johnstone2001,johnstone2004}, as can be seen by the following degrees-of-freedom argument (see also \citet[Section 3.1]{dynamic_cov_est}): as is common when daily or monthly returns are used, the number of historical data points $T$ is of the order of $1000$ while the number of assets $N$ typically ranges between 100 and 1000. Since in this case $T \ll N^2$,  only $O(1)$ effective samples are used to estimate each entry in the covariance matrix, making the sample covariance matrix perform poorly out-of-sample~\citep{LS,NLS,dynamic_cov_est}. 

Given the importance and the statistical challenges of covariance matrix estimation in the high-dimensional setting, this problem has been widely studied in the statistics and financial economics literature. In the statistical literature, a number of estimators have been proposed based on \emph{banding} or \emph{soft-thresholding} the entries of $S$ \citep{bickel2008,pourahmadi_banding,cai2010}. Such estimators, which are equivalent to selecting the covariance matrix closest to $S$ in Frobenius norm subject to the covariance matrix lying within a specified $L_1$ ball, were proven to be \emph{minimax optimal} with respect to the Frobenius norm and spectral norm loss \citep{cai2010}. However, such estimators may not output a covariance matrix estimate that is positive definite, which is required for the Markovitz portfolio selection problem. Moreover, while such estimators are optimal in a minimax sense for the Frobenius and spectral norm loss, these losses may not be relevant to measure the excess risk that results from using an estimate of $\Sigma^*$ instead of $\Sigma^*$ itself to compute the Markovitz portfolio; see \citet[Section 4.1]{dynamic_cov_est} for details. %Hence, the optimality results in \cite{cai2010} may no longer hold for selecting Markovitz portfolios.  %number of papers have proposed different covariance estimators that either exploit some regularity or low-dimensional structure in $\Sigma^*$ for Markowitz portfolio selection; see, for example, \cite{deMiguel2013, Jagannathan2003,fan2008,POET,LS,NLS}. 

Another reason to consider estimators beyond those in \citet{bickel2008,pourahmadi_banding,cai2010} is that these methods do not  exploit some of the structure that often holds in $\Sigma^*$. In particular, the eigenspectrum of $\Sigma^*$ is often structured; we expect to find several important "directions" (i.e., eigenvectors) that well-approximate $S$. For example, under the \emph{capital asset pricing model}~\citep{capm}, the eigenspectrum of $\Sigma^*$ contains a dominant eigenvector % with large eigenvalue 
corresponding to the market; as a consequence, $S$ could be well-approximated by the sum of a rank one matrix (the "market component") and a diagonal matrix (the "idiosyncratic error component"). More generally, covariance matrix estimators based on low-rank approximations of $S$ are advantageous statistically since such estimators have smaller variance% when the rank of the covariance matrix estimator is small
\footnote{If the covariance matrix estimator has rank $M$, then the effective number of parameters estimated is $O(NM)$ instead of $O(N^2)$ where $M \ll N$.}. In practice, low-rank covariance estimates are based on explicitly provided factors \citep{fama_french3,fama_french5,capm}, or data-driven factors learned by performing \emph{principal component analysis} (\emph{PCA}) on $S$ \citep{POET, fan2008}. Another related popular strategy for estimating $\Sigma^*$ is based on the assumption that the eigenvalues of $\Sigma^*$ are well-behaved, and exploit results from random matrix theory~\citep{elkaroui2008,marcenko_dist}. In particular, various methods considered regularizing the eigenvalues of $S$ \citep{LS,NLS,dynamic_cov_est, Jagannathan2003, deMiguel2013}; collectively, these methods can be regarded as particular instances of empirical Bayesian shrinkage estimators~\citep{emp_bayes,LS,stein1956}. Finally, a number of papers have proposed covariance estimators based on the assumption that the precision matrix is \emph{sparse}~\citep{graphical_lasso, RWRY11}. Such a constraint is motivated by the fact that a sparse precision matrix implies that the induced undirected graphical model associated with the joint distribution is sparse, which is desirable both for better interpretability and robustness properties.

In this paper, we propose a new type of covariance matrix estimator for portfolio selection based on the assumption that the underlying distribution is \emph{multivariate totally positive of order 2} ($\text{MTP}_2$), which exploits a particular type of structure in the covariance matrix. $\text{MTP}_2$ was first studied in~\citet{fkg_ineq,karlin_and_rinott,Bolviken,KR83} from a purely theoretical perspective and later also in the context of statistical modeling, in particular graphical models, in~\citet{mle_exists,fallat2017,mle_mtp2,Ising_mtp2}. % that previous methods have not considered. 
$\text{MTP}_2$ is a strong form of positive dependence that can be used in combination with the above methods for covariance estimation. The structure we exploit is motivated by the observation that asset returns are often positively correlated since assets typically move together with the market. As an illustration, consider the sample correlation matrix $S$ and its inverse $S^{-1}$ based on the 2016 monthly returns of global stock markets shown in Figure~\ref{fig:samp_cov_mat}. Note that all correlations (i.e., off-diagonal entries of $S$) are positive, and all partial correlations (i.e., negative of the off-diagonal entries of $S^{-1}$) are positive. %For example, if we look at the sample correlation matrix of the 2016 monthly returns of global stock markets, we see that all entries are positive in \cref{fig:samp_cov_mat}. 
%Moreover, when we invert this matrix, and examine the sample precision matrix, we see that all the off-diagonal entries are negative in \cref{fig:samp_cov_mat}.

A multivariate Gaussian distribution with mean $\mu$ and positive definite covariance matrix $\Sigma$ is $\text{MTP}_2$ if and only if $(\Sigma^{-1})_{ij} \leq 0$ for all $i \neq j$. A precision matrix satisfying this condition is called a \emph{symmetric M-matrix}~\citep{Bolviken,karlin_and_rinott}, and implies that all correlations and partial correlations are non-negative~\citep{Ostrowski,dellacherie2014inverse}. Hence, a multivariate Gaussian  fit to the 2016 daily returns of the global stock market indices considered in Figure~\ref{fig:samp_cov_mat} is $\text{MTP}_2$. This is quite remarkable, since uniformly sampling correlation matrices, e.g.~using the method described in~\citet{Joe}, shows that less than $0.001\%$ of all $5\times 5$ correlation matrices satisfy the $\text{MTP}_2$ constraint. Since factor analysis models with a single factor are $\text{MTP}_2$
when each observed variable has a positive dependence
on the latent factor~\citep{wermuth2014star}, the capital asset pricing model implies $\text{MTP}_2$ when all market betas are positive, which further motivates studying $\text{MTP}_2$ in the context of portfolio selection.

%When a matrix has all negative off-diagonal entries, it is known as an \emph{M-matrix}. When returns are distributed multivariate Gaussian, then the precision matrix being an M-matrix, implies that the joint distribution is $\text{MTP}_2$ \cite{karlin_and_rinott}. In \cref{sec:mult_gaus} we show how to do estimation subject to the constraint that the joint distribution is $\text{MTP}_2$, and further motivate studying $\text{MTP}_2$ distributions in \cref{sec:latent_tree}. \textcolor{red}{Do you think we need more motivation/explanation for MTP2?}

\begin{figure}[!t]
	\centering
	\includegraphics[width=12cm]{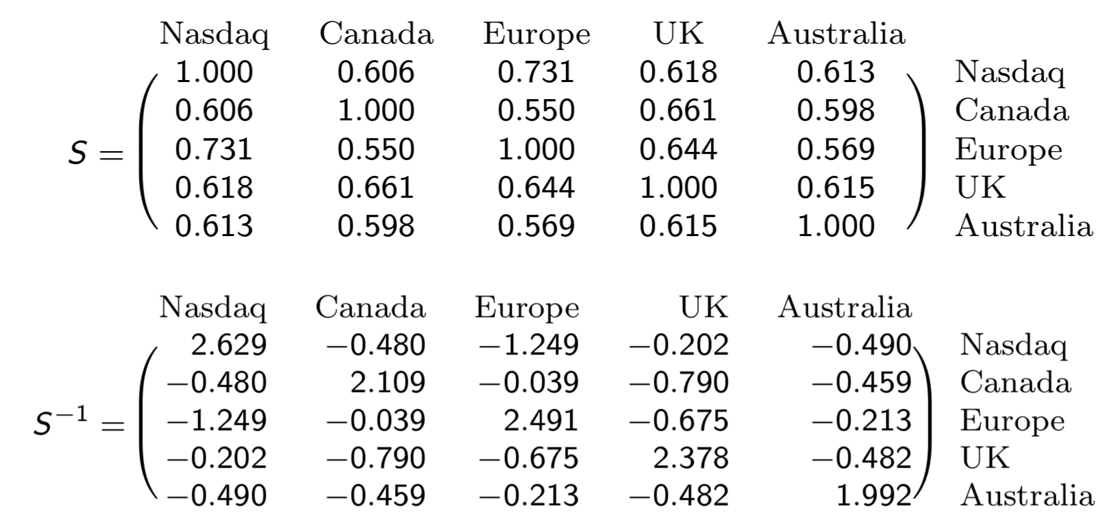}
	\caption{The sample correlation matrix of global stock market indices based on monthly returns from 2013-2016. "Canada", "Europe", "UK", and "Australia" refer to the country names in the MSCI Developed Markets Index}. Notice that the covariance matrix contains all positive entries and the precision matrix is an M-matrix which implies that the joint distribution is $\text{MTP}_2$ (see \cref{sec:mult_gaus} for details).  \label{fig:samp_cov_mat}
\end{figure}

In this paper, we provide (1) a new $\mtp$ covariance matrix estimator to model heavy-tailed returns data and (2) an extensive empirical comparison demonstrating the advantages of this new estimator on stock-market data spanning thirty years. The remainder of this paper is organized as follows: In \cref{sec:background} we review %mathematically detail 
the Markowitz portfolio problem and %review 
existing techniques for covariance matrix estimation that we benchmark our method against in \cref{sec:results}. In \cref{sec:MTP2}, we define $\text{MTP}_2$ more precisely, motivate its usage for financial returns data in more detail, and describe a method to perform covariance estimation under this constraint. Finally, in \cref{sec:results} we empirically compare our method with several competing methods on historical stock market data and show that covariance matrix estimation under $\text{MTP}_2$ outperforms state-of-the-art methods for portfolio selection in terms of out-of-sample variance, i.e.~risk. {All data and code for this work is available at \href{https://github.com/uhlerlab/MTP2-finance}{https://github.com/uhlerlab/MTP2-finance}.}

%% file: 2_background.tex
After introducing some notation, we will review %mathematically detail 
the Markowitz portfolio selection problem, explain how  it relates to covariance matrix estimation, and discuss various covariance estimation techniques. 

\subsection{Notation}
We assume throughout that we are given $N$ assets, which we index using the subscript $i$, from $T$ dates (e.g.~days), which we index using the subscript $t$. %We denote the covariance matrix of a random vector $X$ by $\Cov(X)$.
We let $r_{i, t}$ denote the observed return for asset $i$ at date $t$ for $1 \leq i \leq N$ and $1 \leq t \leq T$. The vector $r_t := (r_{1,t}, \ldots, r_{N, t})^T$ consists of the returns of each asset on day $t$. Finally, $\mu_t := \E[r_t]$ and $\Sigma_t := \Cov(r_t)$ denote the expected returns and the covariance matrix of the returns for day $t$, respectively.

\subsection{Optimal Markowitz Portfolio Allocation}
Markowitz portfolio theory concerns the problem of assigning weights $w\in\mathbb{R}^N$ to a universe of $N$ possible assets in order to minimize the variance of the portfolio for a specified level of expected returns $R$. More precisely, the optimal portfolio weights $w \in \mathbb{R}^N$ on day $t$ are found by solving
%Typically in Markowitz portfolio theory, for a particular date $t$, given knowledge of $\mu^*_t$ and $\Sigma^*_t$ (the true mean and covariance matrix of the assets on day $t$), we want to assign portfolio weights $w \in \mathbb{R}^N$ to our $N$ assets to solve the following optimization problem
%
\begin{equation} \label{eq:full_markowitz}
\begin{aligned}
& \underset{w\in\mathbb{R}^N}{\text{minimize}}
& & w^T \Sigma^*_t w\\
& \text{subject to}
&& w^T\mu_t^* = R \quad \textrm{and}\quad \sum_{i=1}^{N} w_i = 1,
\end{aligned}
\end{equation}
%\begin{equation*}
%\begin{split}
%    \min_{w} w^T \Sigma^*_t w,  \\
%    \text{such that } w^T\mu_t^* = R \\
%    \text{with } \sum_{i=1}^{N} w_i = 1
%\end{split}
%\end{equation*}
where  $\mu^*_t$ and $\Sigma^*_t$ denote the true expected returns and covariance matrix of the returns for day $t$. In practice, $\mu^*_t$ and $\Sigma^*_t$ are unknown and must be estimated from historical returns. Since the main difficulty lies in estimating $\Sigma^*_t$ (it requires estimating $O(N^2)$ parameters as compared to $O(N)$ for  $\mu^*_t$), a widely used tactic to specifically evaluate the quality of a covariance matrix estimator is by finding the \emph{global minimum variance} portfolio, which does not require estimating $\mu^*$  \citep{Haugen35,Jagannathan2003}. Such a portfolio can be found by solving 
%
%for some parameter $R$, which represents the desired level of returns for the portfolio. The interpretation of this objective function is for some target return $R$, we want to assign weights $w$ to assets in our portfolio that minimizes the variance of the portfolio (which is given by $w^T \Sigma_t w$). 
%
%\textcolor{blue}{Although the above formulation of the Markowitz problem has a closed-form solution as a function of $\mu_t$ and $\Sigma_t$, for our purposes } we are not concerned with estimating $\mu_t$ and only concerned with evaluating estimators of $\Sigma_t$. Thus we consider the problem of estimating the global minimum variance (GMV) portfolio. This problem is formulated as 
%
\begin{equation}
\label{eq_a}
\begin{aligned}
& \underset{w\in\mathbb{R}^N}{\text{minimize}}
& & w^T \Sigma^*_t w\\
& \text{subject to}
&& \sum_{i=1}^{N} w_i = 1,
\end{aligned}
\end{equation}
%\begin{equation}
%\label{eq_a}
%\begin{split}
%    \min_w w^T \Sigma_t^* w, \\
%    \text{such that } \sum_{i=1}^{N} w_i = 1,
%\end{split}
%\end{equation}
%
where $w$ is chosen to minimize the variance of the portfolio. Replacing the unknown true covariance matrix of returns $\Sigma_t^*$ by some estimator $\hat{\Sigma}_t$ yields the following analytical solution for \cref{eq_a}:
%We see in the above formulation, given $\mu_t$ and $\Sigma_t$ we are only concerned with assigning portfolio weights $w$ to minimize the variance of the resulting portfolio (without regard to the expected return of the portfolio). This problem has the following analytical solution:
%
%\begin{equation} %\label{eq:markowitz-optimal-weight}
%w = \frac{(\Sigma_t^*)^{-1} \mathbf{1}}{\mathbf{1}^T(\Sigma_t^*)^{-1}\mathbf{1}}
%\end{equation}
%
%where the unknown true covariance matrix Since we do not know the true $\Sigma_t$, in practice, we replace the unknown $\Sigma_t$ by an estimator $\widehat{\Sigma_t}$ in the formula above, yielding a feasible portfolio
%
\begin{equation} \label{weight-estimate-formula}
    \hat{w} := \frac{\hatSigma_t^{-1}\mathbf{1}}{\mathbf{1}^T \hatSigma_t^{-1}\mathbf{1}}.
\end{equation}
%
%Estimating the global minimum variance portfolio is a widely used tactic to evaluate the quality of a covariance matrix estimator \cite{Haugen35,Jagannathan2003}. %, since it removes the need to estimate the vector of expected returns at the same time, which is not the focus of our work. 

A natural choice for $\hatSigma_t$ is the sample covariance matrix. Unfortunately, as discussed in Section~\ref{sec:introduction}, the sample covariance matrix is a poor estimator of the true covariance matrix, particularly in the high-dimensional setting when the number of assets $N$ exceeds the number of periods  $T$ (the sample size). Although the sample covariance matrix is an unbiased estimator of the true covariance matrix, in the high-dimensional setting it is not invertible, has high variance, and is not \emph{consistent} (e.g., the eigenvectors of $S$ do not converge to those of $\Sigma^*$~\citep{marcenko_dist,johnstone2001,wachter1978_issue_samp_cov,bai1993_issue_samp_cov,johnstone2004}). %Furthermore, when $N > T$, where $T$ is the number of periods used to construct the sample covariance matrix (i.e. the sample size), then the sample covariance matrix is not even positive definite. Thus the sample covariance matrix cannot be inverted, which is needed in Formula \ref{weight-estimate-formula}. Although the sample covariance matrix is an unbiased estimator of the true covariance, it has very high variance and in the high-dimensional setting is not even \emph{consistent} (e.g., the eigenvectors of $S$ do not even converge asymptotically to those in $\Sigma^*$ \cite{marcenko_dist,johnstone2001,wachter1978_issue_samp_cov,bai1993_issue_samp_cov,johnstone2004}). 
%To obtain lower variance estimators of the covariance matrix, we make structural assumptions about the true covariance matrix that allows us to construct estimators that exploit this structure and lead to lower variance estimators with only a small increase in bias from the sample covariance matrix. We review several commonly used techniques for covariance matrix estimation in a financial context below.
Making structural assumptions about the true covariance matrix allows the construction of estimators that have lower variance with only a small increase in bias. 

%% file: 3_MTP2.tex
%\textcolor{red}{start with definition of MTP2, %then pos association, then Gaussian, and then cmotivation to finance, which leads over to latent tree models}

%\textcolor{red}{Assets typically positive dependent driven by latent market variable or Fama French factors; MTP2 strongest form of positive dependence, implies e.g. postiive association; motivation to consider this strong form of positive dependence by studying connection to latent tree models in section 3.1; then in section 3.2 discuss how to compute MTP2 covariance estimator in the Gaussian setting and extend the estimator to the non-Gaussian setting in section 3.2;}  

%So far, we have reviewed different methods for high-dimensional covariance matrix estimation assuming a low-rank factor model, regularity in the eigenspectrum, and/or sparsity in the underlying inverse covariance matrix. In the following, 
We propose a new structure for modeling asset returns data, namely by exploiting that assets are often positively dependent. %  motivated by real asset returns data, namely that covariances between the returns of different assets are typically positive; 
In particular, we consider distributions that are $\textit{MTP}_2$. 
\begin{defn}[\cite{fkg_ineq,KarlinRinott80}] %\label{def:mtp2}
A distribution on $\mathcal{X}\subseteq \R^M$ is \emph{multivariate totally positive of order 2} (\emph{$\text{MTP}_2$}) if its density function $p$ satisfies
\begin{equation*} \label{def:mtp2}
p(x)p(y) \leq p(x \wedge y) p(x \vee y) \quad \text{for all} \quad x, y \in \mathcal{X},
\end{equation*}
where $\wedge, \vee$ denote the coordinate-wise minimum and maximum, respectively. %.~\cite{FKG71,KarlinRinott80}.
\end{defn}
$\text{MTP}_2$ is a strong form of positive dependence that implies most other known forms including e.g.~\emph{positive association}; see for example~\citet{Colangelo} for a recent overview. Note that when $p(x)$ is a strictly positive density, then \cref{def:mtp2} is equivalent to $p(x)$ being \emph{log-supermodular}. Log-supermodularity has a long history in ecomomics, in particular in the context of complementarity and comparative statics ~\citep{Topkis78,Milgrom_Roberts,milgram_super,Topkis,athey_super,comp_advan_super}. %, combinatorial optimization \cite{combo_book}, and phylogenetics \cite{semi_alg_stat}. 

In \cref{fig:samp_cov_mat}, we provided an example of 5 global stock indices, where the sample distribution is $\text{MTP}_2$. To further motivate studying $\text{MTP}_2$ as a constraint for covariance matrix estimation for portfolio selection we discuss its connection to latent tree models in Section~\ref{sec:latent_tree}. In particular, we show that the capital asset pricing model implies that the resulting joint distribution is $\text{MTP}_2$ when all "market betas" (also known as "market loadings" or "factor coefficients") are positive. Then in Section~\ref{sec:mult_gaus}, we discuss how to perform covariance matrix estimation under $\text{MTP}_2$ in the Gaussian setting. Finally, in Section~\ref{sec:MTP_heavy}, we propose how to extend this estimator to heavy-tailed distributions.

\subsection{Latent Tree Models} \label{sec:latent_tree}

\begin{figure}[!t]
\centering
\includegraphics[width=12cm]{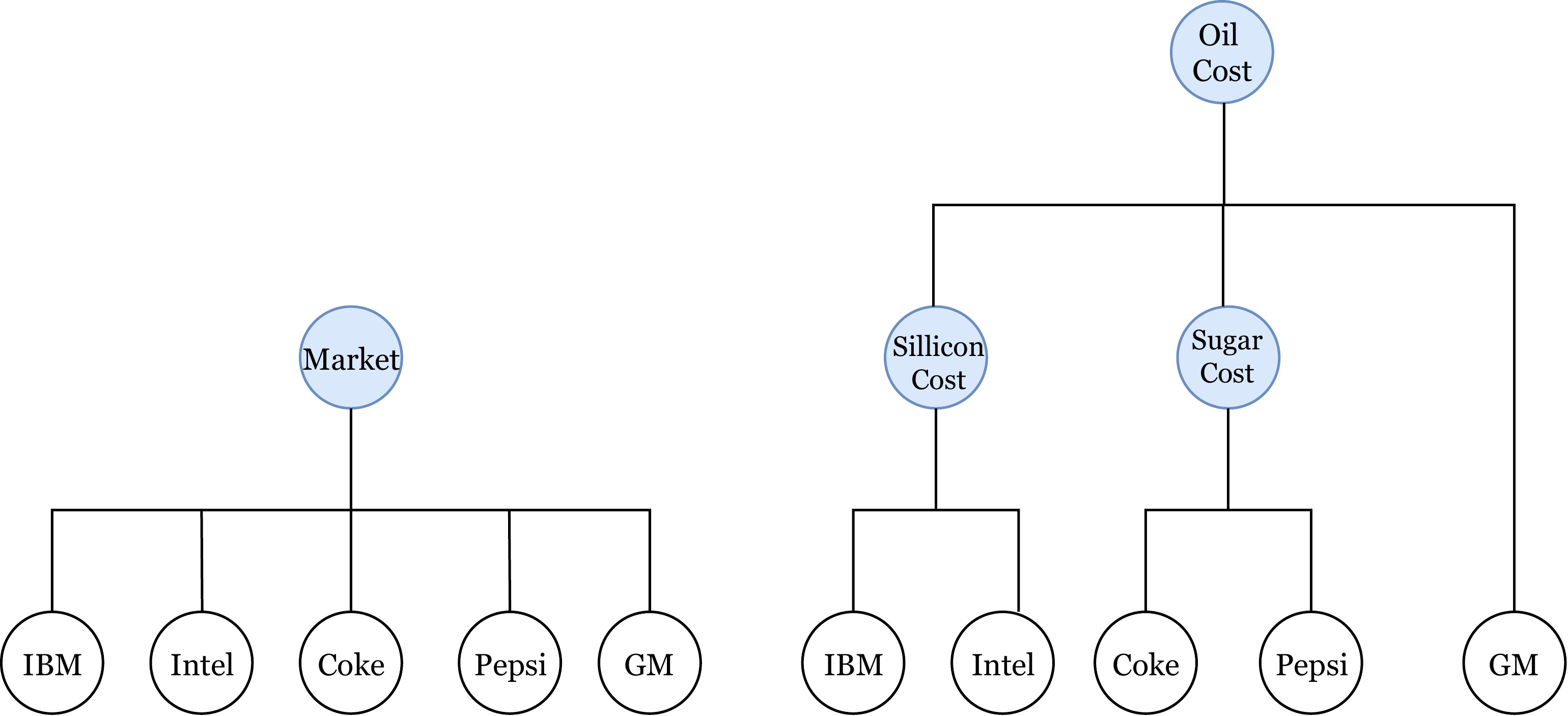}
\caption{Shaded nodes represent factors that are potentially unobserved, and unshaded nodes are the observed returns of different companies. Figure (left) represents a simple model where an unobserved market variable drives the returns of all stocks as in the CAPM. Figure (right) represents a more complicated latent tree model where latent sector-level factors drive the returns of different assets.} \label{fig:tree_models}
\end{figure}

A powerful framework to model complex data such as stock-market returns is through models with latent variables. \emph{Factor models}, which are widely used for covariance estimation for portfolio selection (see Section~\ref{factor_models}) are examples thereof. A \emph{latent tree model} is an undirected graphical model on a tree (where every node represents a random variable that may or may not be observed and any two nodes are connected by a unique path). For financial applications, latent tree models have been used, for example, for unsupervised learning tasks, such as clustering similar stocks, or for modeling and learning the dependence structure among asset returns~\citep{latent_tree_learn,latent_tree_cluster}. A factor analysis model with a single factor is a particular example of a latent tree model consisting of an unobserved root variable that is connected to all the observed variables; see \cref{fig:tree_models} for a concrete example of a single-factor analysis model and a more general latent tree model. The prominent \emph{capital asset pricing model} (\emph{CAPM}) is  a single-factor analysis model: the return of stock $i$ is modeled as
\begin{equation*}
r_i = r_f + \beta_i (r_m - r_f) + u_i \quad \beta_i \in \R,
\end{equation*}
where $r_f$ is known as the risk-free rate of return, $r_m$ is the market return, and $u_i$ is the uncorrelated, zero mean idiosyncratic error term. Typically, the parameters $\beta_i$ are positive, which explains why the covariance between stock returns is usually positive\footnote{Over 97\% of the entries of the sample covariance matrix of 1000 assets (based on daily returns from 1980-2015) are  positive.}. Non-negative correlation is in general necessary but not sufficient to imply $\text{MTP}_2$. The following theorem states that for latent tree models non-negative correlation already implies $\text{MTP}_2$. The proof follows from \citet[Theorem 5.4]{mle_mtp2}.

\begin{thm} \label{thm:tree_mtp2}
Let $X \in \R^M$ follow a  multivariate Gaussian distribution that factorizes according to a tree. %That is, the joint distribution of $X$ factorizes as 
%\begin{equation*}
%p(X) = \prod_{i=1}^M p(X_i \mid \text{Pa}_G(X_i)),
%\end{equation*}
%where $\text{Pa}_G(X_i)$ denotes the parents of node $X_i$ in $G$. 
If $\Cov(X)\geq 0$, then $X$ is $\text{MTP}_2$ and any marginal of $X$ is $\text{MTP}_2$. 
\end{thm}
%
%Any distribution such that $|X|$ is $\text{MTP}_2$ is called \emph{signed} $\text{MTP}_2$. While there are methods to do estimation when $X$ is signed $\text{MTP}_2$ \cite{mle_mtp2}, in this paper we focus on when $X$ is $\text{MTP}_2$ since for latent trees, a covariance matrix with all positive entries implies $X$ is $\text{MTP}_2$. While \cref{thm:tree_mtp2} says that $\text{MTP}_2$ distributions over trees are closed under marginalization, this closure property holds more generally for $\text{MTP}_2$ distributions over arbitrary graphs \cite{karlin_and_rinott}. Hence, even if we do not observe certain factors, if we simply assume that the factors and observed random variables form a latent tree or are $\text{MTP}_2$, then the marginal joint distribution of the observables (e.g., the returns of each asset) is still $\text{MTP}_2$.      
%
%

While working with CAPM is convenient from a theoretical perspective, its simplicity often comes at the expense of underfitting. In particular, there commonly are additional sector-level factors that drive returns. Identifying these factors is an active area of research; for instance, CAPM was recently extended to include three and then five new factors~\cite{fama_french3,fama_french5}. % and recently two more additional factors have been added \cite{}. 
However, identifying relevant factors is in general a challenging task; for example, learning the structure of a latent tree model from data is known to be NP-hard~\citep{latent_np_hard}. We here propose to instead take a \emph{structure-free} approach by constraining the joint distribution over the observed variables to be $\text{MTP}_2$. This approach provides more flexibitlity than modeling stock returns using latent tree models and at the same time allows overcoming the computational bottleneck of fitting a latent tree model. In particular, we show in \cref{sec:mult_gaus} that an $\text{MTP}_2$ covariance matrix estimator can be computed by solving a convex optimization problem.

%Unfortunately, we often do not know the structure of the latent tree, and trying to learn the structure from data is \emph{NP-hard} \cite{latent_np_hard}. Instead, we show in \cref{sec:mult_gaus} how to do covariance estimation subject to an $\text{MTP}_2$ constraint, which again is a necessary but not sufficient condition for the joint distribution of the returns to be generated from a latent tree for the multivariate Gaussian case. By removing the constraint that the joint distribution factorizes according to a latent tree, we show in \cref{sec:mult_gaus} that we can find an $\text{MTP}_2$ covariance estimator by solving a simple \emph{convex} optimization problem. Furthermore, working only under an $\text{MTP}_2$ constraint allows us to model richer graphical structures \cite{fallat2017}. Since the $\text{MTP}_2$ constraint is already quite restrictive, the added variance of modeling graphical structures beyond latent trees is not large, and may in fact be better than the bias of assuming that the joint distribution is generated exactly from a latent tree.

\subsection{$\text{MTP}_2$ Covariance Matrix  Estimation Assuming Multivariate Gaussian Returns} \label{sec:mult_gaus}
For multivariate Gaussian distributions, a necessary and sufficient condition for a distribution to be $\text{MTP}_2$ is that the precision matrix $K \coloneqq \Sigma^{-1}$ is an \emph{M-matrix}, i.e., $K_{ij} \leq 0$  for all $i \neq j$; or equivalently,
all partial correlations are nonnegative. \citep{karlin_and_rinott}. Following \citet{mle_mtp2}, we consider the \emph{maximum likelihood estimator} (\emph{MLE}) of $K$ subject to $K$ being an M-matrix. 

Recall that the log-likelihood function $\mathcal{L}$ of $K$ given data $D \coloneqq \{r_t\}_{t=1}^T \overset{\text{i.i.d.}}{\sim} N(0, K)$ is, up to additive and multiplicative constants, given by
\begin{equation}
\mathcal{L}(K; D) = \log \det K - \text{trace}(KS),
\end{equation}
where $S\in \mathbb{R}^{N \times N}$ denotes the sample covariance matrix of the returns $\{r_t\}_{t=1}^T$ or log-returns.
Without the $\text{MTP}_2$ constraint, the MLE of $K$ is obtained by maximizing $\mathcal{L}(K; D)$ over the set of all positive semidefinite matrices and is given by $S^{-1}$ when $N \leq T$ (i.e., the dimension of the covariance matrix is less than the number of samples). Note that when $N \geq T$, the MLE does not exist, i.e., the log-likelihood function is unbounded above.  %since $S$ must be invertible. 
Remarkably, by adding the constraint that $K$ is an M-matrix (i.e., that the distribution is $\text{MTP}_2$), then the MLE 
\begin{equation} \label{eq:mtp2_gaus}
\hat{K} = \arg \max_{K \succeq 0} \ \log \det K - \text{trace}(KS) \quad \text{subject to} \quad K_{ij} \leq 0 \quad \forall i \neq j,
\end{equation}
exists with probability 1 when $T \geq 2$ for \emph{any} dimension $N$~\cite{mle_exists,mle_mtp2}. Similarly, the popular CLIME estimator, which we review in Eq.~(\ref{eq:clime}) in the next section, could be extended to the MTP$_2$ setting by adding the constraints $K_{ij}\leq 0$ for all $i\neq j$. It would be of interest to understand its properties. 

%we can extend the CLIME estimator to the MTP2 setting by solving 
%
%\begin{equation} \label{eq:clime_mtps}
 %   \hat{K} = \argmin_{K} \| K \|_1 \quad \text{s.t} \quad \|SK - I_N \|_{\infty} \leq \lambda, \quad K_{ij} \leq 0 \quad \forall i \neq j.
%\end{equation}
%

The fact that a unique solution exists for \cref{eq:mtp2_gaus} for any $N$ when $T \geq 2$ suggests that the $\text{MTP}_2$ constraint adds considerable regularization for covariance matrix estimation. In addition, the problem in \cref{eq:mtp2_gaus} %and \cref{eq:clime_mtps} are 
is a convex optimization problem and computationally efficient coordinate-descent algorithms have been described for computing $\hat{K}$~\citep[{cf.}][]{mle_mtp2,mle_exists}. Finally, another desirable property is that the $\text{MTP}_2$ covariance matrix estimator $\hat{K}$ in \cref{eq:mtp2_gaus} is usually \emph{sparse}~\cite[Corallary 2.9]{mle_mtp2}, which reduces the intrinsic dimensionality of the model and hence reduces the variance of the estimator. Note that this sparsity is achieved without the need of any tuning parameter, an immediate advantage over methods that explicitly add sparsity-inducing $L_1$ penalties such as the graphical lasso~\citep{graphical_lasso, RWRY11} discussed in \cref{sec:graph_lasso}. Nevertheless, to relax the $\mtp$ constraint, one could always introduce a Lagrange multiplier (i.e., tuning parameter) to penalize for violating the $\mtp$ constraint\footnote{ Such a strategy can also be used to perform a sensitivity analysis to the $\mtp$ assumption. We thank an anonymous reviewer for raising this point. We leave an empirical evaluation of this strategy to future work.}. %, is that \cref{eq:mtp2_gaus} does not require any tuning parameters to achieve sparsity. 

%We now discuss several desirable properties of the optimization problem in \cref{eq:mtp2_gaus} and \cref{eq:clime_mtps}. The first is that \cref{eq:mtp2_gaus} and \cref{eq:clime_mtps} are convex optimization problems. Hence, computationally efficient \emph{coordinate-descent} solvers can find $\hat{K}$ \cite[Algorithm 1]{mle_mtp2}. Another desirable property is that $\hat{K}_t$ is sparse, namely $\hat{K}_t$ contains many zero entries \cite[Corallary 2.9]{mle_mtp2}. Recall that for multivariate Gaussian distributions, the sparsity pattern of $\hat{K}$ is in 1-1 correspondence with an induced undirected graphical model, namely a graph with edges for all non-zero $\hat{K}_{ij}$ entries. Hence, for high-dimensional problems, having sparsity in $K$ adds regularization and reduces the variance of the estimator because the underlying graphical model is low-dimensional i.e. contains few edges. One immediate advantage of \cref{eq:mtp2_gaus} over methods that explicitly add sparsity-inducing $L_1$ penalties such as the \emph{graphical LASSO} \cite{graphical_lasso} discussed in \cref{sec:graph_lasso}, is that \cref{eq:mtp2_gaus} does not require any tuning parameters to achieve sparsity. 

\subsection{Extensions to Heavy-Tailed Distributions}
\label{sec:MTP_heavy}

%\textcolor{red}{conveniant generalization of gaussian to a class of distributions that includes heavy-tailed distributions such as the t-distribution and are trans-elliptical distributions}

Asset returns are often computed as $r_t = \log\left(\frac{p_t}{p_{t-1}}\right)$, where $p_t$ is the price of the asset at time $t$. Stock returns may be heavy tailed, and in such cases the Gaussian assumption made for estimating the covariance matrix in Section~\ref{sec:mult_gaus} may be problematic. Transelliptical distributions form a convenient class of distributions that contain the Gaussian distribution as well as heavy-tailed distributions such as the $t$-distribution. In the following, we provide an extension of the estimator in  \cref{eq:mtp2_gaus} %and \cref{eq:clime_mtps} 
to transelliptical distributions.

A random vector $X$ with density function $p(x)$, mean $\mu \in \R^M$ and covariance matrix $\Sigma \in \R^{M \times M}$ follows an \emph{elliptical distribution} if its density function can be expressed as  
\begin{equation*} \label{eq:elliptical_g}
   g((x - \mu)^T \Sigma^{-1} (x - \mu))
\end{equation*}
for some function $g$.
%
%\begin{equation*}
%X \overset{\text{d}}{=} \mu + \xi \cdot A \cdot U, 
%\end{equation*}
%
%where $AA^T = \Sigma$ is the Cholesky decomposition of $\Sigma$, $U \in R^M$ is a unit vector drawn uniformly at random from the unit sphere, and $\xi \geq 0$ is a scalar random variable that is independent of $U$. The multivariate Gaussian case occurs when $\xi$ follows a chi-squared distribution with $M$ degrees-of-freedom. 
More generally, $X$ follows a \emph{transelliptical distribution} if there exist monotonically increasing functions $f_i$, $i=1, \dots , M$, such that $(f_1(X_1), \cdots, f_M(X_M))$ follows an elliptical distribution. We denote the covariance matrix of this elliptical distribution by $\Sigma_f$. The following result provides a necessary condition for a transelliptical distribution to be $\text{MTP}_2$.

 \begin{thm} \label{thm:trans_m_mat} Suppose that the joint distribution of $(X_1, \cdots, X_M)$ is $\text{MTP}_2$ and transelliptical, i.e., there exist  increasing functions $f_i$, $i=1, \dots , M$, such that the density function of $(f_1(X_1), \cdots, f_M(X_M))$ can be written as $g((x - \mu)^T \Sigma_f^{-1} (x - \mu))$. Then, $\Sigma_f^{-1}$ is an M-matrix.
 \end{thm}
 
We prove \cref{thm:trans_m_mat} in \ref{sec:proofs}. While \cref{thm:trans_m_mat} shows that the covariance matrix of any elliptical distribution is an inverse M-matrix, the following example shows that, unlike in the Gaussian setting, this is not a sufficient condition for $\text{MTP}_2$.

 \begin{ex}
 Suppose $X$ is a two-dimensional  $t$-distribution with one degree of freedom and precision matrix
 \begin{equation*}
     \Sigma^{-1} = \begin{bmatrix}
    1 & -0.1 \\
    -0.1 & 1 \\
\end{bmatrix}.
 \end{equation*}
 Then $X$ is not $\text{MTP}_2$, since for $x = (-1, 1)$ and $y = (0, 0)$ its density function $p(\cdot)$ satisfies $p(x) p(y) > p(x \wedge y) p(x \vee y)$. \hfill \qed
\end{ex}

This shows that for transelliptical distributions, the constraint that $\Sigma^{-1}$ be an M-matrix is a relaxation of $\text{MTP}_2$. In terms of covariance matrix estimators for transelliptical distributions (without the  $\text{MTP}_2$ constraint), it was shown recently that replacing the sample covariance matrix $S$ in \cref{eq:graph_lasso} and \cref{eq:clime} by \emph{Kendall's tau correlation matrix} $S_\tau$ defined in \cref{Kendall_tau} yields consistent estimators of $\Sigma_f$~\citep{Liu_transellipticalgraphical,rocket}. This is quite remarkable, since it does not involve any changes to the objective function apart from replacing $S$ by $S_\tau$. Motivated by these results, we propose to extend the $\text{MTP}_2$ covariance matrix estimator from Section~\ref{sec:mult_gaus} to heavy-tailed distributions using the covariance matrix estimator in \cref{eq:mtp2_gaus} %or \cref{eq:clime_mtps} 
 by simply replacing the sample covariance matrix $S$~by~$S_\tau$. 
 
 \ra{In recent work, \citet{piotr_mtp2}  provide a number of interesting theoretical results for transelliptical distributions, including the theoretical analysis of our proposed $\text{MTP}_2$ relaxation above. They show that our relaxation for transelliptical distributions has a number of desirable properties, including positive partial correlations for arbitrary conditioning sets and the avoidance of Simpsons Paradox; see \citep[Proposition 4.12]{piotr_mtp2} for details. \citet{piotr_mtp2} further motivate this relaxation by showing that $\text{MTP}_2$ is in fact too strong of a constraint for (non-Gaussian) transelliptical distributions in Theorem 4.8 (for example, there does not exist \emph{any} transelliptical $\text{MTP}_2$ $t$-distributions).}

%% file: related_work.tex
In this section, we review several models and techniques for covariance matrix estimation that are commonly used in financial contexts. We compare our method to these estimators in  \cref{sec:results}.

\subsection{Factor Models} \label{factor_models}
A common modeling assumption in financial applications is that the returns for day~$t$ are given by a linear combination of a (small) collection of latent factors $f_{k, t}$ for $1 \leq k \leq K$, which are either explicitly provided or estimated from the data. In such a factor model, the returns are modeled as
%behind a class of covariance matrix estimators is the \emph{factor model}. In a factor model, there there are $K$ factors indexed by the subscript $k$ where $f_{k,t}$ for $1 \leq k \leq K$ and $1 \leq t \leq T$ denotes the observed return for factor $k$ at date $t$. We let $f_t := (f_{1,t}, \ldots, f_{K, t})$ denote the stacked vector of all observed factor returns for day $t$. As factor model assumes that 
\begin{equation}
    \label{eq_b}
r_{i,t} = \alpha_i + \beta_i^T f_t + u_{i, t}, \quad f_t:= (f_{1,t}, \ldots, f_{K, t}),
\end{equation}
where $u_{i,t}$ is the idiosyncratic error term for asset $i$ that is uncorrelated with $f_t$. Letting $B\in \mathbb{R}^{K \times N}$ be the matrix whose $i$th column is $\beta_i$, the covariance matrix of the returns can be expressed as 
\begin{equation*} \label{eq:factor_cov_mat}
    \Sigma_{t} = B^T \Sigma_{f, t} B + \Sigma_{u, t}, \quad \text{for} \quad 1 \leq t \leq T,
\end{equation*}
where $\Sigma_{f,t} := \Cov(f_t)$ and $\Sigma_{u, t} := \Cov(u_t)$. In practice, $K \ll N$ factors are selected, making $B^T \Sigma_{f, t} B$ low-rank. % latent factors drive the majority of the observed correlation between assets, 
This low-rank structure makes estimating $\Sigma_{t}$ easier since $\Sigma_{f, t}$ and $B$ only have $O(K^2)$ and $O(NK)$ free parameters, respectively. 
%
%Together, these two assumptions make estimating $ \Sigma_{r, t}$ easier since each covariance matrix term in \cref{eq:factor_cov_mat} is relatively easy to estimate, as we argue below. 
%
%$B^T \Sigma_{f, t} B$ requires estimating $\Sigma_{f, t}$, which has $O(K^2)$ free parameters, and $B$, which has $O(NK)$ free parameters. Again, the factors are either explicitly provided or estimated from the data (see bullets below); however their covariances may not be known. 
When $K \ll N$, and $K \ll T^2$, then by standard concentration of measure results, $\Sigma_{f, t}$ can be estimated well by $\hatSigma_{f, t}$, the sample covariance matrix of the factors. Similarly, by \cref{eq_b}, the $i$th row of $B$ can be estimated by regressing the returns of asset $i$ on the $K$ latent factors, for example using ordinary least-squares. %Again, when $K$ is small relative to $N$ and $T$, these regression coefficients can be estimated well from standard estimators such as ordinary least-squares. 
%If the regression coefficient estimate 
In this case, $\hat{\beta}_i \approx \beta_i$ and hence the error $u_{i, t}$ is approximately equal to the residual $\hat{u}_{i, t} \coloneqq r_{i, t} - \hat{\beta}_i^T f_t - \hat{\alpha}_i$. Thus $\Sigma_{u, t}$ can be approximated by a covariance matrix estimate $\hat{\Sigma}_{u, t}$ based on the residuals. However, without additional assumptions on the structure of $\Sigma_{u, t}$, $\Sigma_{u, t}$ is not necessarily easier to estimate than $\Sigma_t$.   As a result, many estimators assume that $\Sigma_{u, t}$ has some special structure such as being diagonal or sparse (see below).

%\textcolor{red}{The following sentences are not clear. Please explain better how the estimation works. In particular, why is it easier to estimate $\Sigma_{u,t}$ as compared to $\Sigma_{r,t}$? How do you get $\hat\beta$? (I assume regression) Please explain.} Most factor-model based covariance estimators derive estimates $\hatSigma_{f, t}$ and $\hatSigma_{u, t}$ for $\Sigma_{f, t}$ and $\Sigma_{u, t}$ respectively. Then the estimator for $\Sigma_{r, t}$ is given by
%$$\hatSigma_{r, t} = \hat{B}^T \hatSigma_{f, t} \hat{B} + \hatSigma_{u, t}.$$
%Intuitively, because the number of factors $K$ is generally small (and certainly much smaller than the number of assets $N$, so $K << N$), it is easier to obtain better estimates for $\Sigma_{f,t}$ and $\Sigma_{u,t}$ for the same number of samples, as compared to estimates for $\Sigma_{r,t}$, as the former matrices have smaller dimension than the latter.

Several different types of factor models of varying complexity  have been considered in the literature: The general model in~\cref{eq_b} is known as a \emph{dynamic factor model}. A \emph{static factor model} assumes that the covariance matrices $\Sigma_{u,t}$ and $\Sigma_{f, t}$ are time-invariant, i.e., $\Sigma_{u,t}=\Sigma_u$ and $\Sigma_{f, t}=\Sigma_f$ do not depend on $t$. An \emph{exact factor model} furthermore assumes that the covariance matrix $\Sigma_{u}$ is diagonal, whereas an \emph{approximate factor model} assumes that $\Sigma_u$ has bounded $L^1$ or $L^2$ norm. In this paper, we concentrate on static estimators. The following static factor-based covariance matrix estimators are popularly used in financial applications. 
%
%\begin{itemize}
%    \item A \emph{static factor model} assumes that the covariance matrices $\Sigma_{u,t}$ and $\Sigma_{f, t}$ are time-invariant, i.e. they do not depend on $t$.
%    \item An \emph{exact factor model} assumes that the covariance matrix $\Sigma_{u}$ is diagonal; an \emph{approximate factor model} assumes that $\Sigma_u$ has bounded $L^1$ or $L^2$ norm.
%    \item A \emph{dynamic factor model} relaxes the assumption that the covariance matrices $\Sigma_{u}, \Sigma_f$ are time-invariant and allows these matrices to depend on the date $t$.
%\end{itemize}
%
%
%
%\subsection{Covariance Estimators for Factor Models}
%
\begin{itemize}
    \item \textbf{POET: } is based on an approximate factor model and was first proposed in \cite{POET}. POET estimates $B^T \Sigma_{f, t} B$ by a rank $K$ truncated singular value decomposition (SVD) of the sample covariance matrix $\hat{\Sigma}$, which we denote by $\hat{\Sigma}_K$.  $\hatSigma_u$ is estimated by soft-thresholding the off-diagonal entries of the residual covariance matrix $S_{\hat{u}}  = \hat{\Sigma} - \hat{\Sigma}_K$ based on the method in \cite{bickel2008}.
    %The factors are given by the top principal components of $\{ r_t \}$ and the number of principal components is determined in a data-driven manner. The estimate of the factor covariance matrix $\hatSigma_f$ is given by the sample covariance matrix of the determined principal components. The estimate $\hatSigma_u$ is obtained by applying thresholding to the off-diagonal entries of  $S_{\hat{u}}$, %where $S_{\hat{u}}$ denotes 
    %the sample covariance matrix of the residuals~$\{ \hat{u}_t\}$. %This is a static estimator, since the assumption is a static factor model. 
    \item \textbf{EFM: } is an estimator based on the exact factor model using the Fama-French factors \citep{fama-french}. $\hatSigma_f$ equals the sample covariance matrix of the factors $\{ f_t \}$ and $\hatSigma_u$ equals the diagonal of $S_{\hat{u}}$. 
    \item \textbf{AFM-POET: } is an estimator based on an approximate factor model using the Fama-French factors. $\hatSigma_f$ is obtained as in EFM, whereas $\hatSigma_u$ is obtained by soft-thresholding $S_{\hat{u}}$ as in POET. %The thresholding scheme involves applying the POET method to $\{ \hat{u} \}$ where the number of principal components is set to $0$ and the POET estimate of $\hatSigma_u$ is used. %This is a static estimator as well.
\end{itemize}

%All of the estimators above are all static estimators and our estimator is also a static estimator. There have been several dynamic factor-model based estimators that account for the time-varying nature of $\Sigma_f$ and $\Sigma_u$, which is more realistic. However, as our novel estimator is (currently) a static estimator, we want to benchmark it against other static estimators.

\subsection{Shrinkage of Eigenvalues} \label{sec:shrinkage}

Another way to impose structure on the covariance matrix is through assumptions on the eigenvalues of the covariance matrix. Assuming that the true covariance matrix is well-conditioned, then the extreme eigenvalues of the sample covariance matrix are generally too small/large as compared to the true covariance matrix% assumptions are more general than the factor model structure discussed above. Broadly speaking, the idea behind these assumptions is that the extreme eigenvalues of the sample covariance matrix are generally too large or too small
~\citep{marcenko_dist, bai1993_issue_samp_cov}. This motivates the development of covariance matrix estimators such as \emph{linear shrinkage}~\cite{LS} and extensions thereof ~\citep[{cf.}][]{NLS,dynamic_cov_est} that shrink the eigenvalues of the sample covariance matrix for better statistical properties. % The solution is to assume that the eigenvalues of the true covariance matrix are more well-behaved, and to \emph{shrink} the eigenvalues of the sample covariance matrix away from being too small or too large. The primary method that we review here is \emph{linear shrinkage} proposed in \cite{LS}; for extensions see \textcolor{red}{add refs!}

%\subsubsection{Linear Shrinkage}
%As proposed in \cite{LS}, 
To be more precise, let  
$$ S = \sum_{i=1}^{N} \lambda_i v_i v_i^T,$$ 
be the eigendecomposition of the sample covariance matrix $S$, where $\lambda_i$ denotes the $i$-th eigenvalue of $S$ and $v_i$ the corresponding eigenvector. Then the linear shrinkage estimator is given by
\begin{equation*} \label{eq:LS_formula}
\hatSigma_{LS} = \sum_{i=1}^{N} \gamma_i v_i v_i^T,
\end{equation*}
where $\gamma_i = \rho \lambda_i + (1-\rho) \bar{\lambda}$ with $\bar{\lambda}$ denoting the average of the eigenvalues of $S$ and $0<\rho<1$ a tuning parameter that determines the amount of shrinkage. Note that $\hatSigma_{LS}$ can equivalently be expressed as 
\begin{equation} \label{eq:LS_formula_2}
\hatSigma_{LS} = \rho S + (1-\rho) \bar{\lambda} I_N,
\end{equation}
where $I_N\in\mathbb{R}^{N\times N}$ denotes the identity matrix (Eq.~(\ref{eq:LS_formula_2}) follows from the uniqueness of the eigenvalue decomposition). % and noting that $\rho S + (1-\rho) \bar{\lambda} I_N$ are given by $\lambda_i$ and $\gamma_i$. %the eigenvalues and eigenvectors of  This equality can be derived from noticing that the $v_i$ are eigenvectors of the matrix expression on the right-hand side of Equation~\ref{eq:LS_formula_2} with eigenvalues given exactly by $\gamma_i$. 
Thus $\hatSigma_{LS}$ is obtained by shrinking the sample covariance matrix towards a multiple of the identity, which from a Bayesian point of view can also be interpreted as using the identity matrix as a prior for the true covariance matrix~\citep{LS}. % $I_N$ as a \emph{prior} from a Bayesian point-of-view
The shrinkage estimator $\hatSigma_{LS}$ is asymptotically efficient given a particular choice of $\rho$ that depends on the sample covariance matrix $S$, its dimension $N$ (i.e., the number of assets) and the number of samples $T$ (i.e., the number of dates)~\citep{LS}. 

An extension of linear shrinkage, known as \emph{non-linear shrinkage}, considers non-linear transforms of the eigenvalues according to the Marcenko-Pastur distribution, which describes the asymptotic distribution of the eigenvalues of random matrices. This approach has been shown to %relationship between population and sample eigenvalues. Empirically, the non-linear shrinkage estimator 
outperform linear-shrinkage  empirically~\citep{NLS}. It is also common to combine shrinkage estimators with factor models (e.g., such as those introduced in \cref{factor_models}). 
%These shrinkage estimators can also be combined with the factor model estimators from Section~\ref{factor_models} by applying linear or non-linear shrinkage to the residual covariance matrix $S_{\hat{u}}$ to estimate $\hatSigma_u$. 
For example, \emph{AFM-LS} and \emph{AFM-NLS} apply linear shrinkage and non-linear shrinkage, respectively, to the residuals (by regressing out the Fama-French factors) to estimate $\Sigma_u$ \citep{factormodelworkingpaper}.

\subsection{ Regularization of the Precision Matrix} \label{sec:graph_lasso}

Another common technique for covariance matrix estimation is to assume that the true underlying inverse covariance matrix $K^*:=(\hat{\Sigma}^*)^{-1}$, also known as the \emph{precision matrix}, is sparse, i.e.~that the number of non-zero entries in $K^*$ is bounded by an integer $\kappa>0$.
%
%\begin{equation} \label{eq:sparse_prec}
%    \|K^*\|_0 \coloneqq \sum_{i \neq j} I(K^*_{ij} \neq 0) \leq \kappa% \quad \text{s.t.} \quad K \coloneqq \Sigma^{-1},
%\end{equation}
%
%where $\kappa>0$ is an integer that controls the sparsity level of $K$. 
Since estimating $K$ under the constraint
\begin{equation} \label{eq:sparse_prec}
    \|K\|_0 \coloneqq \sum_{i \neq j} I[K_{ij} \neq 0] \leq \kappa% \quad \text{s.t.} \quad K \coloneqq \Sigma^{-1},
\end{equation}
is computationally intractable as it involves solving a difficult combinatorial optimization problem, a standard approach is to replace the $L_0$ constraint in \cref{eq:sparse_prec} by an $L_1$ constraint. In particular, assuming that the data follows a multivariate Gaussian distribution, then the $L_1$-regularized maximum likelihood estimator (also known as \emph{graphical lasso}) can be used to estimate $K$~\citep{graphical_lasso, RWRY11}. Maximum likelihood estimation under the  the $L_1$ constraint leads to the following convex optimization problem:%  estimator for $K$ is given by solving the convex optimization problem,
\begin{equation} \label{eq:graph_lasso}
    \hat{K} \coloneqq \argmax_{K \succeq 0} \ \log \det K - \text{trace}(KS) \quad \text{subject to} \quad  \|K\|_1 \leq \lambda,
\end{equation}
%
%where the objective function is the log-likelihood of the data, which is assumed to be distributed as a multivariate Gaussian, and $\lambda$ is a hyperparameter controlling the sparsity level of $\hat{K}$ such that for larger $\lambda$, $\hat{K}$ becomes sparser. 
where $\lambda \geq 0$ is a tuning parameter. %\textcolor{green}{It is important to note that graphical lasso is only a consistent estimator of the zero-pattern of the precision matrix $K$ (for a particular unknown choice of $\lambda$ that depends on $n,p$ and properties of the underlying true precision matrix). It is \emph{not} a consistent estimator of the non-zero entries of $K$. To obtain a consistent estimator for such a quantity, typically graphical lasso is used to obtain the $0$-pattern, and then the maximum-likelihood objective is refit once again, subject to this $0$-pattern. However, standard implementations of graphical lasso, such as the one in sklearn, do not do this refitting and instead use the biased entries of the original fitting procedure. Thus in our experiments, we use the standard implementation of graphical lasso as well and not that it is biased.} 
Instead of maximizing the log-likelihood, the popular \emph{CLIME} estimator \cite{Liu_transellipticalgraphical} finds a sparse estimate of the precision matrix by solving
\begin{equation} \label{eq:clime}
    \hat{K} \coloneqq \argmin_{K} \| K \|_1 \quad \text{subject to} \quad \|SK - I_N \|_{\infty} \leq \lambda.
\end{equation}
and has similar consistency guarantees as the graphical lasso in the Gaussian setting. 

To overcome the restrictive Gaussian assumption, %seems restrictive, interestingly 
recent work %showed that by 
suggested replacing the sample covariance matrix $S$ in \cref{eq:graph_lasso} and \cref{eq:clime} by \emph{Kendall's tau} correlation matrix $S_\tau$ with 
$(S_\tau)_{ij}\coloneqq\sin (\frac{\pi}{2} \hat{\tau})$, where
 \begin{equation}\label{Kendall_tau}
    \hat{\tau}_{ij} \coloneqq \frac{1}{{T \choose 2}} \sum_{1 \leq t \leq t^\prime \leq T} \text{sign}(X_{it} - X_{it^{\prime}})\ \text{sign}(X_{jt} - X_{jt^{\prime}}).
\end{equation}
Interestingly, the resulting estimators can also be used for data from heavy-tailed distributions (including elliptical distributions such as the $t$-distribution) with almost no loss in efficiency~\citep{Liu_transellipticalgraphical,rocket}; see also Section~\ref{sec:MTP_heavy}.

%% file: 4_results.tex
In this section, we first describe both the data used for the evaluation and our experimental setup, which closely follows \citet{factormodelworkingpaper} for reproducibility. We then present our empirical evaluation of the various methods discussed in this paper based on the global minimum variance portfolio problem and the full Markovitz portfolio problem. {All data and code for this work is available at \href{https://github.com/uhlerlab/MTP2-finance}{https://github.com/uhlerlab/MTP2-finance}.}

\subsection{Data} \label{sec:methodology}

We use daily stock returns data %\footnote{We thank Michael Wolf for providing us with a clean, easy-to-use dataset.} 
from the Center for Research in Security Prices (CRSP), starting in 1975 and ending in 2015. We restrict our attention to stocks from the NYSE, AMEX and NASDAQ stock exchanges, and consider different portfolio sizes $N \in \{100, 200, 500\}$. As in \citet{factormodelworkingpaper}, $21$ consecutive trading days constitute one `month'. To account for distribution shift over time, we use a rolling out-of-sample estimator. That is, for each month in the out-of-sample period, we estimate the covariance matrix using the most recent $T$ daily returns, and update the portfolio monthly. We vary $T$ with $N$ to evaluate how sensitive different covariance estimators are with respect to increasing dimensionality. In particular, for a given $N$, we vary $T$ such that the ratio $N/T \in \{ \frac{1}{2}, 1, 2, 4\}$. We also include $T=1260$ (which corresponds to 5 years of market data) in order to replicate the results in \citet{factormodelworkingpaper}. We consider 360 months for evaluation, starting from  01/08/1986 and ending on 12/02/2015, using the portfolio and covariance updating strategy described above. We index each of these 360 investment periods by $h \in \{1, \ldots, 360 \}$. 

For each investment period and portfolio size, we vary the investment universe because many stocks do not have data for the entire period and the most relevant stocks (i.e. by market capitalization or volume) naturally vary over time. We use the same procedure as in \citet{factormodelworkingpaper} to construct the investment universe. Specifically, we consider the set of stocks that have (1) an almost complete return history over the most recent $T = 1260$ days and (2)  a complete return `future' in the next 21 days (which is the investment period). Next, we remove one stock in each pair of highly correlated stocks, defined as those with sample correlation exceeding $0.95$. More precisely, for each pair we remove the stock with the lower market capitalization for period $h$. Finally, we pick the largest $N$ stocks (as measured by their market capitalization on the investment date $h$) for the subsequent analysis. We use $I_{h, N}$ to denote this investment universe, where the subscripts emphasize the dependence on $N$ and $h$. %This methodology is largely similar to the methodology used for evaluation in \cite{factormodelworkingpaper}.

\subsection{Competing Covariance Matrix Estimators} 

We compare the performance of the proposed MTP$_2$ covariance matrix estimator to the estimators described in \cref{sec:past_cov_estimations}. In addition, as a baseline, we also consider the equally weighted portfolio denoted by $\mathbf{1 / N}$. We evaluate each estimator in terms of its out-of-sample standard deviation (see \cref{sec:sd_metric}), Sharpe ratio (see \cref{sec:info_metric}), and information ratio (see \ref{A:add_figs}). These results are also summarized in \cref{table:STD} and \cref{table:sharpe}. %We describe both the evaluation metrics and results in greater depth below. 
In the following, we provide details regarding the implementation of the various covariance matrix estimators included in our empirical analysis.

%We describe each of the methods we compare against more specifically below. 

\begin{itemize}
    \item \textbf{LS: } linear shrinkage, as described in Section \ref{sec:shrinkage}, applied to the sample covariance matrix.
    \item \textbf{NLS: } non-linear shrinkage, as described in Section \ref{sec:shrinkage}, applied to the sample covariance matrix; we used the implementation in the \verb+R+ package \verb+shrink+ \citep{shrink_R_package}.
    \item \textbf{AFM-LS: } approximate factor model, as described in Section \ref{factor_models}, with $5$ Fama-French factors and linear shrinkage applied to estimate the covariance matrix of the residuals.
    \item \textbf{AFM-NLS: } approximate factor model, as described in Section \ref{factor_models}, with $5$ Fama-French factors and non-linear shrinkage applied to estimate the covariance matrix of the residuals.
    \item \textbf{POET (k=3): } POET, as described in Section \ref{factor_models}, using the top $3$ principal components; %We do not use the data-driven criterion as described in Section \ref{factor_models} to choose the number of components because of its computational expense. 
    we used the implementation in the \verb+R+ package \verb+POET+.
    \item \textbf{POET (k=5): } POET, as described in Section \ref{factor_models}, using the top $5$ principal components; we used the implementation in the \verb+R+ package \verb+POET+. %We use POET with the top $5$ principal components.
    \item \textbf{GLASSO: }  graphical lasso, as described in Section \ref{sec:graph_lasso}, using the \verb+python+ implementation in \verb+sklearn+ \citep{sklearn}; cross-validation is used to select the hyperparameter $\lambda$; we used the default parameters, i.e.~using 3-fold cross-validation and testing $\lambda$ on a grid of 4 points refined 4 times (the parameter values for $\alpha$ and $n_\text{iter}$ respectively). We note that this results in a biased estimator due to the $\ell_1$-penalty.%We note that typically graphical lasso is only used for structure recovery (i.e. where the 0's of the covariance matrix are located).
    \item \textbf{CLIME: } as described in Section \ref{sec:graph_lasso}; we used the implementation in the \verb+R+ package \verb+CLIME+ with hyperparameter $\lambda = \sqrt{(\log p)/n}$, which is asymptotically optimal; the CLIME estimator using this hyperparameter only exists when $T \geq N$ and hence we only benchmarked CLIME in this range. % as the sample covariance matrix is singular. As a result, for these scenarios, we are not able to benchmark the performance of CLIME.
        \item \textbf{CLIME-KT: } CLIME estimator as described above but using Kendall's tau correlation matrix instead of the sample correlation matrix. Since Kendall’s tau correlation matrix is not singular, the CLIME-KT estimator exists even when
 when $T \leq N$.
        \item \textbf{MTP2: } our method, as described in Section~\ref{sec:mult_gaus}. We used the implementation from~\citet{mle_exists}, which is a computationally efficient coordinate-descent algorithm implemented in \verb+Matlab+\footnote{The implementation can be found at \url{https://sites.google.com/site/slawskimartin/code.}}.
         \item \textbf{MTP2-KT: } MTP2 estimator as described above but using Kendall's tau correlation matrix instead of the sample correlation matrix; see Section~\ref{sec:MTP_heavy}.
\end{itemize}

\subsection{Evaluation on the Global Minimum Variance Portfolio Problem} \label{sec:sd_metric}

%Recall from Section \ref{sec:background} that we evaluate different estimators on the problem of estimating the global minimum variance (GMV) portfolio. Recall that there are 360 investment periods indexed by $h$, $N$ denotes the number of assets in our portfolio and $T$ denotes the amount of previous daily returns used for covariance estimation (i.e. the sample size).
For each fixed portfolio size $N$, estimation sample size $T$, and investment period $h$, we let $\hatSigma^{\mathcal{M}}_{T, h}(I_{h, N})$ denote the estimated covariance matrix between the assets in universe $I_{h, N}$ obtained using estimator $\mathcal{M}$. We then computed the portfolio weights $\hat{w}^{\mathcal{M}}_{h}$ via \cref{weight-estimate-formula} and the corresponding returns $r^{\mathcal{M}}_h$ for $h = 1, \dots, 360$. 
%To perform the empirical comparison of different covariance estimation algorithms, for a particular fixed $N, T$ for each investment period $h = 1, \ldots, 360$, we take the relevant investment universe (dependant on $h, N$), which we denote as $I_{h, N}$ and estimate the covariance $\hatSigma^{\mathcal{M}}_{h}(I_{h, N})$,  where $\mathcal{M}$ symbolizes the particular method that we're using. Then we compute the portfolio weights $\hat{w}^{\mathcal{M}}_{h}$ using Equation \ref{weight-estimate-formula} from $\hatSigma^{\mathcal{M}}_{h}(I_{h, N})$. Then we compute the returns generated by the portfolio with weight $\hat{w}^{\mathcal{M}}_{h}$ during the investment period $h$, which we denote as $r^{\mathcal{M}}_h$. 
%After this procedure, for a particular method (and fixed $N, P$), we have a resulting set of returns for each investment period: $$\{ r^{\mathcal{M}}_{h} \}_{1 \leq h \leq 360}.$$ 
We estimated the portfolio standard deviation from these 360 returns for each estimator and multiplied each standard deviation by $\sqrt{12}$ to annualize. Note that a smaller standard deviation implies a lower variance portfolio, and hence better empirical performance. %Since we are optimizing for minimum variance and not expected return relative to risk, we do not benchmark using the information ratio. 

\input{sd_table.tex}

\cref{table:STD} summarizes the results for each estimator. Each row corresponds to a particular choice of $N$ (size of investment universe) and $T$ (estimation sample size). Each column corresponds to a different covariance matrix estimator. The best performing estimator in each row is marked with an asterisk. While no estimator outperforms all other estimators across all $N$ and $T$, \cref{table:STD} shows that the MTP2, non-linear shrinkage (NLS), and POET estimators perform consistently well in all settings. 

As discussed in Section \ref{sec:MTP_heavy}, to deal with the heavy-tailed nature of the distribution of returns,   Kendall's tau correlation matrix can be used instead of the sample correlation matrix in the CLIME and MTP2 estimators which assume Gaussianity. Columns CLIME-KT and MTP2-KT in Table~\ref{table:STD}  indicate that while using Kendall's tau correlation matrix usually does not make a significant difference in the performance, it can 
give a slight boost for the MTP2 estimator in particular when $N$ is 100 or 200. 
%.in CLIME does not make a significant difference in the performance of the estimator (in fact it sometimes slightly increases the variance), using Kendall's tau covariance matrix instead of the sample covariance matrix is advisable with the MTP2 estimator. 

%The last two columns in Table \ref{table:STD} show that using Kendall's tau covariance matrix in the estimators (denoted \textbf{MTP2-KT} and \textbf{CLIME-KT}) does not make a significant difference in the performance. Note that since Kendall's tau covariance matrix is not singular, the CLIME-KT estimator exists even when $T\leq N$.

Instead of comparing the covariance matrix estimators only based on one number, the standard deviation of the returns of the resulting portfolios across the entire out-of-sample period, it is also of interest to examine the performance of each estimator \emph{throughout} the out-of-sample period. Figure~\ref{fig:lines} shows the standard deviation of the returns of the different estimators for $N \in \{100, 200,500\}$ and $T=1260$ when varying the out-of-sample period from 60 to $360$ (where $360$ is the maximal number of total out-of-sample months). Note that the ordering between the different estimators is  relatively consistent over time, indicating that the conclusions from the comparison of the different estimators in \cref{table:STD} would remain unchanged even when varying the  length of the out-of-sample period.

\begin{figure}
    \centering
    \includegraphics[width=1.0\textwidth]{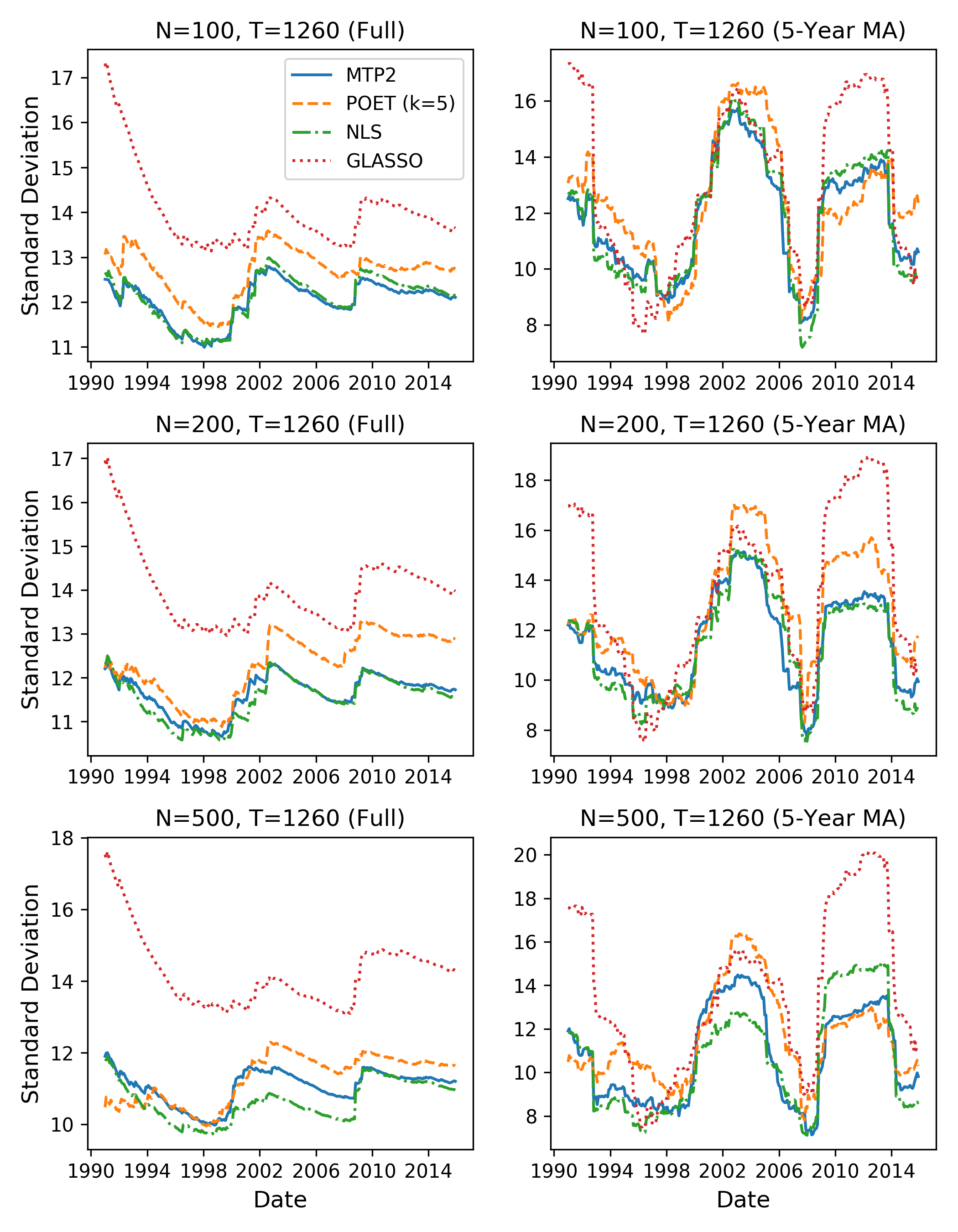}
    \vspace{-0.7cm}
    \caption{By varying the length of the out-of-sample period we examine the standard deviation of the returns obtained by each estimator throughout time. "Full" is the cumulative average while "5-Year MA" is a 5 year moving average. Lower is better.}
    \label{fig:lines}
\end{figure}

\subsection{Evaluation on Full Markowitz Portfolio Problem with Momentum Signal} \label{sec:info_metric}
We also benchmarked the different covariance matrix estimators based on the performance of the portfolios selected by solving \cref{eq:full_markowitz}, where $\Sigma^*_t$ is replaced by the  estimator. A standard performance metric is the \emph{Sharpe ratio}, which is the ratio between the  excess portfolio returns and the standard deviation of excess returns\footnote{We use 1 Year US Treasury Rates to compute the risk-free rate.}. Hence, a higher Sharpe ratio indicates better performance.

%which is the ratio between the expected portfolio returns $R$ and the standard deviation (i.e., risk) of the portfolio.

\input{sharpe_table.tex}

We selected the desired expected returns level $R$ as in \citet{factormodelworkingpaper}. Namely, we considered the \emph{EW-TQ} portfolio which places equal weight on each of the top $20\%$ of assets (based on expected returns). We then set $R$ equal to the expected return of the EW-TQ portfolio. In addition, since the true vector of expected returns $\mu^*$ is unknown, we estimated it from the data. We do this using the momentum factor \citep{momentum} as in \citet{factormodelworkingpaper}, which for a given investment period $h$ and stock is the geometric average of returns of the previous year excluding the past month. 

%We let the momentum signal be denoted by $m$, which is obtained by collecting all the individual momentum signals for each stock. 

%With a predictive signal $m$ and covariance estimate $\hatSigma$, the optimal weights for the portfolio are given by
%$$ \hat{w} := c_1 \hatSigma^{-1} \mathbf{1} + c_2\hatSigma^{-1}m,$$
%for some constants $c_1$ and $c_2$ given in closed form as a function of $m$ and $\hatSigma$ (can be found explicitly in \cite{factormodelworkingpaper}).

%In the context of the Markowitz problem with a predictive signal, the most important performance measure is no longer the standard deviation of returns, since we're no longer trying to construct the minimal variance portfolio. Instead, in this context, the most important metric is the out of sample \emph{information ratio} (i.e. the ratio of the average out of sample return to the out of sample standard deviation). 

The out-of-sample Sharpe ratio and information ratio of each estimator are shown in \cref{table:sharpe} and \ref{A:add_figs}, respectively. %Since it was observed in \cref{table:STD} that using Kendall's tau covariance matrix instead of the sample covariance matrix in the CLIME and MTP2 estimators makes almost no difference, we omitted CLIME-KT and MTP2-KT in \cref{table:IR}. 
As in \cref{table:STD}, each row corresponds to a different choice of $N$ and $T$ and each column corresponds to a different estimator for both tables. The best performing estimator in each row is marked with an asterisk. This analysis shows that the MTP2 estimator achieves the best performance for almost all choices of $N$ and $T$. Although the results are similar, comparing MTP2 to MTP2-KT indicates that it is recommended to use Kendall's tau correlation matrix instead of the sample correlation matrix with the MTP2 estimator when $N$ is 100 or 200.  %\textcolor{purple}{\textbf{Refer to the information ratio results in the appendix!}}\ra{I now added the reference to the appendix at the start of this paragraph since the same analysis holds for both tables}

Similar to Figure~\ref{fig:lines}, in Figure~\ref{fig:lines_2} we show the Sharpe ratio of the returns of the different estimators for $N \in \{100, 200,500\}$ and $T=1260$ when varying the out-of-sample period from 60 to $360$. Note that while the ordering between the different estimators is  still relatively consistent over time, it varies more than for the standard deviation plotted in Figure~\ref{fig:lines} and could provide additional valuable information regarding each estimator that is not captured in \cref{table:sharpe}.

%For this experiment to see whether the $\text{MTP}_2$ estimator leads to higher out of sample information ratios, we compare it against the linear shrinkage estimator and approximate factor model (as before). As a naive baseline, we introduce the \emph{EW-TQ} portfolio which takes the top $20\%$ of assets (according to the momentum signal) and places equal weight amongst them. The results of comparing the out of sample information ratios for the various methods and various choices of $N, T$ are in Table \ref{table:IR}.
%

\begin{figure}[!t]
    \centering
    \includegraphics[width=1.0\textwidth]{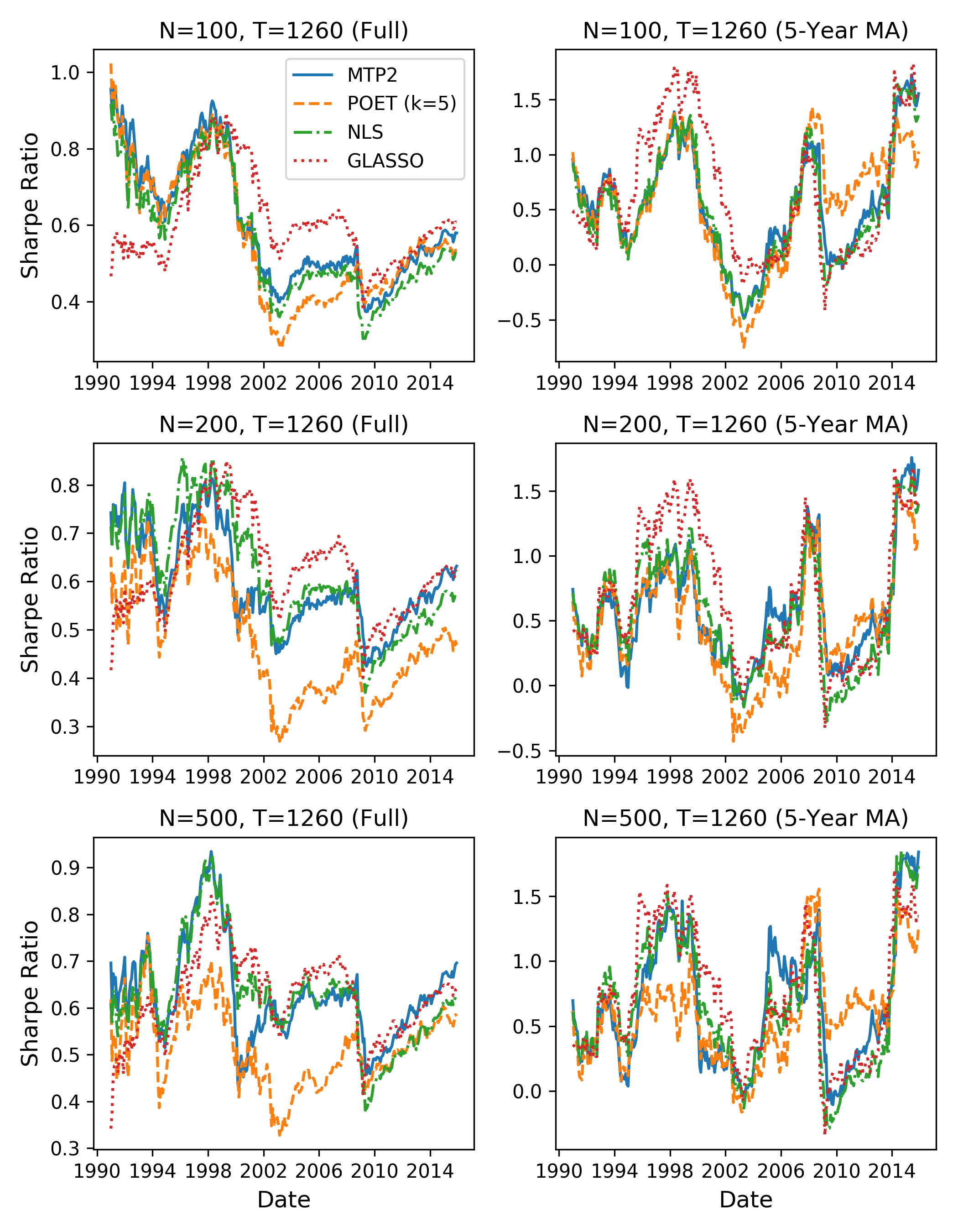}
    \vspace{-0.7cm}
    \caption{By varying the length of the out-of-sample period we examine the Sharpe ratio of the returns obtained by each estimator throughout time. "Full" is the cumulative average while "5-Year MA" is a 5 year moving average. Higher is better.}
    \label{fig:lines_2}
\end{figure}

%% file: sd_table.tex
\begin{table}[b!]
\centering
\begin{tabular}{cc|ccccccc}
\toprule
N    &  T    &     1/N &      LS &                       NLS &                    AFM- &                   AFM- &                POET &                POET \\
&  &         &         &                           &               LS            &                     NLS      &            (k=3)               &       (k=5)                    \\
\midrule
100 & 50   &  18.724 &  13.452 &                    12.976 &                    13.159 &                    13.193 &  {12.498}* &                    12.617 \\
    & 100  &  18.724 &  13.695 &                    13.111 &                    13.135 &                    13.338 &  {11.994}* &                    12.595 \\
    & 200  &  18.724 &  12.560 &                    12.347 &                    12.357 &                    12.480 &                    12.348 &                    12.707 \\
    & 400  &  18.724 &  12.451 &                    12.347 &                    12.352 &                    12.344 &                    12.744 &                    13.255 \\
    & 1260 &  18.724 &  12.151 &                    12.122 &                    12.146 &                    12.130 &                    13.041 &                    12.722 \\
200 & 100  &  18.134 &  12.583 &                    12.320 &                    12.372 &                    12.406 &                    11.743 &  11.544 \\
    & 200  &  18.134 &  11.881 &                    11.603 &  11.556 &                    11.612 &                    11.881 &                    11.593 \\
    & 400  &  18.134 &  11.656 &  {11.431}* &                    11.552 &                    11.469 &                    12.559 &                    12.103 \\
    & 800  &  18.134 &  11.670 &  {11.424}* &                    11.531 &                    11.449 &                    13.019 &                    12.455 \\
    & 1260 &  18.134 &  11.665 &  {11.534}* &                    11.601 &                    11.568 &                    13.170 &                    12.898 \\
500 & 250  &  17.925 &  11.140 &                    10.516 &                    10.508 &                    10.517 &                    11.269 &  {10.203}* \\
    & 500  &  17.925 &  11.934 &  {10.793}* &                    10.913 &                    11.163 &                    11.833 &                    10.873 \\
    & 1000 &  17.925 &  11.373 &                    10.838 &                    10.856 &  {10.816}* &                    12.179 &                    11.917 \\
    & 1260 &  17.925 &  11.469 &  {10.943}* &                    11.005 &                    10.950 &                    12.395 &                    11.626 \\
\bottomrule
\end{tabular}

\bigskip

\begin{tabular}{cc|ccccc}
\toprule
 N   &  T    &                        GLASSO &   CLIME & CLIME-& MTP2 & MTP2- \\
 &  &                           &       &  KT &&  KT    \\
\midrule
100 & 50   &                     13.594 &     nan & 15.484 & 12.655 & 12.623  \\
   & 100&                       13.822 &     nan & 15.024& 12.327 & 12.049  \\
    & 200  &    13.985 &  14.945 & 15.140 & 11.858 & {11.742}*  \\
    & 400  &    13.607 &  15.127 & 15.223& 12.294 & {12.114}*  \\
    & 1260 &    13.631 &  15.253 & 15.316& {12.087}* & {12.087}*  \\
200 & 100  &                      13.522 &     nan & 14.983& 11.803 & {11.445}*  \\
    & 200  &                      13.719 &     nan & 14.344& 11.586 & {11.442}*  \\
    & 400  &                      13.920 &  14.563 & 14.964& 11.880 & 11.905  \\
    & 800  &                      14.096 &  14.778 & 14.862 & 11.635 & 11.661  \\
    & 1260 &                      13.958 &  15.013 & 15.013& 11.710 & 11.749  \\
500 & 250  &                      13.855 &     nan & 15.677& 10.455 & 10.512  \\
    & 500  &                      14.171 &     nan & 20.896& 11.009 &  11.261   \\
    & 1000 &                      14.283 &  15.523 & 14.330& 11.031 & 11.273  \\
    & 1260 &                      14.290 &  14.776 & 14.962& 11.187 &  11.422    \\
\bottomrule
\end{tabular}\qquad\qquad\quad
\caption{For each combination of $N$ (portfolio size), $T$ (estimation sample size), and covariance matrix estimator, we report the out-of-sample standard deviation of the returns of the portfolio. The most competitive value in each row is marked with an asterisk.
} \label{table:STD}
\end{table}

\begin{comment}
\begin{tabular}{cc|cccc|cc}
\toprule
 N   &  T    &                      MTP2 &  GLASSO &   CLIME && MTP2- & CLIME-\\
 &  &                           &       &  &     & KT & KT    \\
\midrule
100 & 50   &                    12.655 &  13.594 &     nan & & 12.623 & 15.484 \\
    & 100  &                    12.327 &  13.822 &     nan & & 12.049 & 15.024 \\
    & 200  &  \textcolor{blue}{11.858} &  13.985 &  14.945 & & 11.742 & 15.140 \\
    & 400  &  \textcolor{blue}{12.294} &  13.607 &  15.127 & & 12.114 & 15.223 \\
    & 1260 &  \textcolor{blue}{12.087} &  13.631 &  15.253 & & 12.087 & 15.316 \\
200 & 100  &                    11.803 &  13.522 &     nan & & 11.445 & 14.983 \\
    & 200  &                    11.586 &  13.719 &     nan & & 11.442 & 14.344 \\
    & 400  &                    11.880 &  13.920 &  14.563 & & 11.905 & 14.964 \\
    & 800  &                    11.635 &  14.096 &  14.778 & & 11.661 & 14.862 \\
    & 1260 &                    11.710 &  13.958 &  15.013 & & 11.749 & 15.013 \\
500 & 250  &                    10.455 &  13.855 &     nan & & 10.512 & 15.677 \\
    & 500  &                    11.009 &  14.171 &     nan & &  --.-- & --.--  \\
    & 1000 &                    11.031 &  14.283 &  15.523 & & 11.273 & 14.330 \\
    & 1260 &                    11.187 &  14.290 &  14.776 & &  --.-- & --.--   \\
    \bottomrule
\end{tabular}
\end{comment}

%% file: sharpe_table.tex
\begin{table}[b!]
\centering
\begin{tabular}{cc|ccccccc}
\toprule
N    & T     &  EQ-TW &     LS &    NLS &                   AFM- & AFM- & POET & POET \\
 &  &        &        &        &   LS                       &    NLS     &  (k=3)          &    (k=5)        \\
\midrule
100 & 50   &  0.544 &  0.348 &  0.361 &  0.334 &   0.338 &      0.462 &      0.496 \\
    & 100  &  0.544 &  0.328 &  0.397 &  0.344 &   0.340 &      0.486 &      0.394 \\
    & 200  &  0.544 &  0.374 &  0.419 &  0.389 &   0.376 &      0.500 &      0.413 \\
    & 400  &  0.544 &  0.437 &  0.471 &  0.502 &   0.475 &      0.532 &      0.474 \\
    & 1260 &  0.544 &  0.525 &  0.527 &  0.526 &   0.524 &      0.555 &      0.539 \\
200 & 100  &  0.599 &  0.423 &  0.433 &  0.413 &   0.428 &      0.448 &      0.439 \\
    & 200  &  0.599 &  0.498 &  0.471 &  0.474 &   0.468 &      0.432 &      0.443 \\
    & 400  &  0.599 &  0.545 &  0.559 &  0.566 &   0.568 &      0.528 &      0.513 \\
    & 800  &  0.599 &  0.649 &  0.636 &  0.640 &   0.643 &      0.461 &      0.571 \\
    & 1260 &  0.599 &  0.588 &  0.585 &  0.593 &   0.585 &      0.491 &      0.481 \\
500 & 250  &  0.599 &  0.649 &  0.639 &  0.641 &   0.638 &      0.538 &      0.664 \\
    & 500  &  0.599 &  0.628 &  0.609 &  0.653 &   0.668 &      0.534 &      0.685 \\
    & 1000 &  0.599 &  0.592 &  0.633 &  0.650 &   0.636 &      0.470 &      0.550 \\
    & 1260 &  0.599 &  0.595 &  0.628 &  0.646 &   0.642 &      0.505 &      0.589 \\
\bottomrule
\end{tabular}

\medskip

\begin{tabular}{cc|ccccc}
\toprule
 N   &  T    &                        GLASSO &   CLIME & CLIME-&MTP2 & MTP2- \\
 &  &                           &       & KT &     & KT    \\
 
 \midrule
100 & 50   &     0.589 &       nan &     0.548 &     0.554 &  {0.611}* \\
    & 100  &     0.616 &       nan &     0.589 &     0.594 &  {0.666}* \\
    & 200  &     0.589 &     0.580 &  {0.636}* &     0.585 &     0.634 \\
    & 400  &     0.603 &     0.608 &     0.578 &     0.590 &  {0.617}* \\
    & 1260 &  {0.605}* &     0.535 &     0.523 &     0.582 &     0.547 \\
200 & 100  &  {0.611}* &       nan &     0.593 &     0.514 &     0.594 \\
    & 200  &     0.587 &       nan &  {0.632}* &     0.563 &     0.594 \\
    & 400  &     0.597 &  {0.657}* &     0.568 &     0.573 &     0.581 \\
    & 800  &     0.596 &     0.605 &     0.552 &  {0.650}* &     0.627 \\
    & 1260 &     0.620 &     0.593 &     0.632 &  {0.638}* &     0.615 \\
500 & 250  &     0.639 &       nan &     0.341 &     0.755 &  {0.779}* \\
    & 500  &     0.623 &       nan &     0.313 &  {0.705}* &     0.674 \\
    & 1000 &     0.637 &     0.572 &  {0.818}* &     0.723 &     0.635 \\
    & 1260 &     0.635 &     0.585 &    0.539  &  {0.701}* &     0.635 \\
\bottomrule
\end{tabular}\qquad\qquad\;
\vspace{-0.2cm}
\caption{For each combination of $N$ (portfolio size), $T$ (estimation sample size), and covariance matrix estimator, we report the out-of-sample Sharpe ratio (the ratio between the  excess portfolio returns and the standard deviation of excess returns based on 1 Year US Treasury Rates).  The most competitive value in each row is marked with an asterisk.
%Performance of different estimators for various $N$ (number of assets) and $T$ (number of samples considered). For each combination of $N, T$ and choice of estimator, the number in the table is the information ratio (ratio of the average return to the standard deviation of returns) of the portfolio given by applying the estimator to the rolling out of sample.
} \label{table:sharpe}
\end{table}

%% file: 5_conclusion.tex
In this paper, we proposed a new covariance matrix estimator for portfolio selection based on the assumption that returns are $\text{MTP}_2$, which is a strong form of positive dependence. While the $\text{MTP}_2$ assumption is strong, this constraint adds considerable regularization, thereby reducing the variance of the resulting covariance matrix estimator. Empirically, the added bias of $\text{MTP}_2$ is outweighed by the reduction in variance. In particular, the proposed $\text{MTP}_2$ estimator outperforms previous state-of-the-art covariance matrix estimators in terms of the Sharpe ratio and the information ratio.

In our empirical evaluation we observed that using Kendall tau's correlation matrix instead of the sample covariance matrix in the MLE under MTP$_2$ performed particularly well for a portfolio size of 100 or 200. It would therefore be of interest to analyze the theoretical properties of such covariance matrix estimators including MLE or CLIME under MTP$_2$ for heavy-tailed distributions. In addition, while we only considered  static covariance matrix estimators in this paper,  the $\text{MTP}_2$ estimator naturally extends to the dynamic setting, where the covariance matrix evolves over time. Specifically, we may adapt the techniques developed in \citet{dynamic_cov_est} to obtain a dynamic estimator under $\text{MTP}_2$. In future work, it would be interesting to compare the resulting estimator to other state-of-the-art dynamic covariance matrix estimators. Another interesting future direction is the theoretical analysis of the spectrum of symmetric M-matrices in the high-dimensional setting. % would be to theoretically analyze the (high-dimensional) asymptotic behavior of the spectrum of symmetric M-matrices. 
If the $\text{MTP}_2$ constraint already implicitly regularizes the spectrum sufficiently, then shrinkage methods such as those developed in \citet{LS,NLS,dynamic_cov_est, Jagannathan2003, deMiguel2013} may be unnecessary under $\text{MTP}_2$. Alternatively, covariance matrix estimators under $\text{MTP}_2$ could be combined with shrinkage methods to potentially achieve even better performance.

%% file: appendix.tex
\section{Proofs} \label{sec:proofs}
The proof of \cref{thm:trans_m_mat} requires the following simple lemma.
 
 \begin{lem} \label{lem:helper}
 Suppose $g(x)$ is differentiable, non-negative, and $\int_{-\infty}^{\infty} g(x) dx = 1$. Then, for any $\delta, M > 0$, there exists an $x^* > M$ such that $g(\cdot)$ is strictly decreasing on the interval $(x^*, x^* + \delta)$.
 \end{lem}
 \begin{proof} Let $I = \{x: g^{\prime}(x) > 0 \}$. Then, the Lebesgue measure of $I$ is finite since $g(\cdot)$ is non-negative and integrates to one. %Hence, $\lambda(I^c \cap \{x: x > M\}) = \infty$. 
 Suppose towards a contradiction that there was no such $x^*$. Then, for any $x > M$, $g(\cdot)$ is not  monotonically decreasing on $(x, x+\delta)$. Hence, by continuity of $g(\cdot)$, there exists an interval $I_x$ of length $\Delta_x$ contained in $(x, x+\delta)$ such that $g(\cdot)$ is monotonically increasing on $I_x$. Let $\bigsqcup_{j=1}^{\infty} I_{x_j}$ be some disjoint covering of $\{x: x > M\}$, where $I_{x_j} \coloneqq (x_j, x_j + \delta]$. Then, by our previous argument, $I_{x_j}$ contains an interval of length $\Delta_{x_j}$ where $g(\cdot)$ is monotonically increasing. By assumption, $\inf_{j} \Delta_{x_j} > 0$ and $\liminf_{j \rightarrow \infty} \Delta_{x_j} > 0$. Hence, $\sum_{j} \Delta_{x_j} = \infty$ which contradicts that $I$ has finite Lebesgue measure.
 \end{proof}

 \begin{proof}[Proof of \cref{thm:trans_m_mat}]
 Note that by \citet[Equation 1.13]{karlin_and_rinott}, if $X$ is $\text{MTP}_2$, then so is $(f_1(X_1), \cdots, f_M(X_M))$. Hence $\Sigma_{ij}\geq 0$ for all $i\neq j$. To complete the proof, we need to show that $(\Sigma^{-1})_{ij}\leq 0$ for all $i\neq j$. Without loss of generality, we assume that $\mu = 0$. We consider the two points $x = s_1 e_i - s_2 e_j$ and $y = -x$, where $e_k\in\mathbb{R}^M$ denotes the $k$-th unit vector and $s_i \in \R$. For ease of notation, let $\Sigma^{-1}_{i,i} = a, \Sigma^{-1}_{j,j} = b$, and $\Sigma^{-1}_{i,j} = \Sigma^{-1}_{j,i} = c$. %To complete the proof, we need to show that $c\leq 0$. 
 Notice that
 $$p(x) = p(y) = g(s_1^2 a + s_2^2 b - 2s_1 s_2c) \quad \text{and} \quad p(x \vee y) = (x \wedge y) = g(s_1^2a + s_2^2b + 2s_1s_2c).$$
 Hence, since $(f_1(X_1), \cdots, f_M(X_M))$ is $\text{MTP}_2$, it holds that
 \begin{equation*} \label{eq:proof_mtp2_step1}
    g(s_1^2 a + s_2^2 b - 2s_1 s_2c)^2 \leq g(s_1^2 a + s_2^2 b + 2s_1 s_2c)^2,
\end{equation*}
which simplifies to 
$g(s_1^2 a + s_2^2 b - 2s_1 s_2c)\leq g(s_1^2 a + s_2^2 b + 2s_1 s_2c)$. 
% Now, let $A$ and $B$ be two random variables with covariance equal to $c$ and variances equal to $a$ and $b$, respectively. 
% %
% \begin{equation}
% \begin{split}
%       0 \leq \Cov(s_1 A - s_2 B, s_1 A - s_2 B) & = s_1^2 a + s_2^2 b - 2s_1 s_2c \\
%       0 \leq \Cov(s_1 A + s_2 B, s_1 A + s_2 B) & = s_1^2 a + s_2^2 b + 2s_1 s_2c.
% \end{split}
% \end{equation}
% %
Let $s_2 = \frac{1}{s_1}$ and $\delta = 4|c|$. If $c=0$, the claim trivially holds. Therefore, suppose $|c| > 0$. Then, \cref{lem:helper} implies that there exists an $x^*$ such that $g(\cdot)$ is monotonically decreasing on $(x^*, x^* + 4|c|)$. Since the range of the function $h(s) = as^2 + \frac{b}{s^2}$ is $(M, \infty)$ for some $M > 0$, then by \cref{lem:helper} there must exist $s_1 \in \R$ such that $x^* = s_1^2 a + \frac{b}{s_1^2}$. Since $g(x^* - 2c) \leq g(x^* + 2c)$, then
\begin{equation*}
    x^* - 2c \geq x^* + 2c
\end{equation*}
by monotonicity, which implies $c < 0$ as desired. 
\end{proof}

\section{Information Ratio Results} \label{A:add_figs}

In \cref{sec:info_metric}, we compared the methods in terms of the Sharpe ratio. Here, we provide similar results except for the information ratio, which is the ratio between the expected portfolio returns and portfolio standard deviation.

\input{ir_table_appendix.tex}

%% file: ir_table_appendix.tex
\begin{table}[b!]
\centering
\begin{tabular}{cc|ccccccc}
\toprule
N    & T     &  EQ-TW &     LS &    NLS &                   AFM- & AFM- & POET & POET \\
 &  &        &        &        &   LS                       &    NLS     &  (k=3)          &    (k=5)        \\
\midrule
100 & 50   &  0.694 &  0.625 &  0.648 &                    0.617 &   0.621 &      0.760 &      0.791 \\
    & 100  &  0.694 &  0.600 &  0.682 &                    0.628 &   0.620 &      0.797 &      0.690 \\
    & 200  &  0.694 &  0.670 &  0.720 &                    0.691 &   0.675 &      0.802 &      0.706 \\
    & 400  &  0.694 &  0.736 &  0.772 &                    0.803 &   0.776 &      0.824 &      0.753 \\
    & 1260 &  0.694 &  0.831 &  0.834 &                    0.832 &   0.831 &      0.841 &      0.831 \\
200 & 100  &  0.757 &  0.719 &  0.735 &                    0.715 &   0.728 &      0.766 &      0.762 \\
    & 200  &  0.757 &  0.812 &  0.793 &                    0.796 &   0.790 &      0.747 &      0.764 \\
    & 400  &  0.757 &  0.864 &  0.885 &                    0.888 &   0.892 &      0.825 &      0.820 \\
    & 800  &  0.757 &  0.967 &  0.961 &                    0.962 &   0.967 &      0.747 &      0.870 \\
    & 1260 &  0.757 &  0.906 &  0.907 &                    0.913 &   0.906 &      0.773 &      0.770 \\
500 & 250  &  0.764 &  0.985 &  0.995 &                    0.997 &   0.993 &      0.869 &      1.030 \\
    & 500  &  0.764 &  0.940 &  0.955 &                    0.995 &   1.003 &      0.849 &      1.027 \\
    & 1000 &  0.764 &  0.918 &  0.976 &                    0.993 &   0.980 &      0.772 &      0.861 \\
    & 1260 &  0.764 &  0.920 &  0.967 &  0.984 &   0.982 &      0.806 &      0.909 \\
\bottomrule
\end{tabular}

\bigskip

\begin{tabular}{cc|ccccc}
\toprule
 N   &  T    &                        GLASSO &   CLIME & CLIME-&MTP2 & MTP2- \\
 &  &                           &       & KT &     & KT    \\
 
 \midrule
    100 & 50   &                      0.858 &                      nan & 0.788&0.849 &  {0.905}*  \\
    & 100  &                      0.885 &                      nan & 0.837&0.896 & {0.975}*   \\
    & 200  &                      0.855 &                    0.830 & 0.882&0.899 &  {0.950}*  \\
    & 400  &                      0.877 &                    0.852 &  0.823&0.892 & {0.924}*  \\
    & 1260 &                      0.878 &                    0.778 & 0.767&{0.890}* &  0.855 \\
200 & 100  &                      0.887 &                      nan & 0.844&0.829 & {0.918}*   \\
    & 200  &                       0.859 &                      nan & 0.896&0.885&  {0.919}* \\
    & 400  &                                        0.865 &  {0.916}* & 0.821&0.886 &  0.893  \\
    & 800  &                      0.862 &                    0.860 & 0.805&{0.970}* &  0.945 \\
    & 1260 &                      0.887 &                    0.845 & 0.885&{0.955}* &  0.931  \\
500 & 250  &                      0.908 &                      nan &  0.596&1.112 & {1.133}*  \\
    & 500  &                      0.887 &                      nan & 0.511 &{1.045}* & 1.005 \\
    & 1000 &                      0.897 &                    0.828 &  {1.101}*&1.061 & 0.993 \\
    & 1260 &                      0.896 &                    0.858 & 0.806 &{1.034}* & 0.958 \\
\bottomrule
\end{tabular}\qquad\qquad\;
%\vspace{-0.2cm}
\caption{For each combination of $N$ (portfolio size), $T$ (estimation sample size), and covariance matrix estimator, we report the out-of-sample information ratio (ratio of the average return to the standard deviation of return) of the portfolio.  The most competitive value in each row is marked with an asterisk.
%Performance of different estimators for various $N$ (number of assets) and $T$ (number of samples considered). For each combination of $N, T$ and choice of estimator, the number in the table is the information ratio (ratio of the average return to the standard deviation of returns) of the portfolio given by applying the estimator to the rolling out of sample.
} \label{table:IR_appendix}
\end{table}